\documentclass[sigconf]{acmart}

\renewcommand\footnotetextcopyrightpermission[1]{} %
\usepackage{snapshot}
\usepackage{paralist}
\usepackage{color}
\usepackage{bm}
\usepackage{mdframed}
\usepackage{soul}
\usepackage{tcolorbox}
\usepackage{graphicx}
\usepackage{grffile}

\usepackage{url}
\usepackage{xspace}
\usepackage[mathscr]{euscript}
\usepackage{subfig}
\usepackage{booktabs}
\usepackage{graphbox}

\usepackage{amsfonts}

\usepackage{amsmath}

\usepackage{footnote}
\makesavenoteenv{tabular}

\usepackage{amsthm}
\usepackage[linesnumbered,ruled,vlined]{algorithm2e}
\usepackage{algpseudocode}

\usepackage{longtable}

\makeatletter
\newtheorem*{rep@theorem}{\rep@title}
\newcommand{\newreptheorem}[2]{%
\newenvironment{rep#1}[1]{%
 \def\rep@title{#2 \ref{##1}}%
 \begin{rep@theorem}}%
 {\end{rep@theorem}}}
\makeatother

\allowdisplaybreaks

\usepackage{ifpdf}

\ifpdf
 
\fi

\theoremstyle{plain}
\newtheorem{theorem}{Theorem}[section]%

\newtheorem{lemma}[theorem]{Lemma}
\newtheorem{corollary}[theorem]{Corollary}
\newtheorem{claim}[theorem]{Claim}
\newtheorem{assumption}[theorem]{Assumption}
\newtheorem{definition}[theorem]{Definition}

\theoremstyle{definition}
\usepackage{footnote}

\newcommand{\vol}{\mathrm{vol}}

\newcommand{\poly}{\mathrm{poly}}

\newcommand{\R}{\mathbb{R}}

\newcommand{\1}{\mathbf{1}}

\renewcommand{\qed}{\nobreak \ifvmode \relax \else
	\ifdim\lastskip<1.5em \hskip-\lastskip
	\hskip1.5em plus0em minus0.5em \fi \nobreak
	\vrule height0.75em width0.5em depth0.25em\fi}

\newcommand{\junk}[1]{{}}

\providecommand{\abs}[1]{\lvert#1\rvert}
\providecommand{\lrabs}[1]{\left\lvert#1\right\rvert}
\providecommand{\norm}[1]{\lVert#1\rVert}

\providecommand{\abs}[1]{\left|#1\right|}
\providecommand{\norm}[1]{\lVert#1\rVert}

\providecommand{\mat}[1]{\ensuremath\mathbf#1}

\usepackage{enumitem}

\begin{document}
 \fancyhead{}

	\title{Local Algorithms for Estimating Effective Resistance}

\author{Pan Peng}
\affiliation{%
	\institution{Department of Computer Science\\University of Sheffield}
		\country{}
}
\email{p.peng@sheffield.ac.uk}
\orcid{0000-0003-2700-5699}

\author{Daniel Lopatta}
\affiliation{%
	\institution{Department of Computer Science\\University of Sheffield}
		\country{}
}
\email{dlopatta1@sheffield.ac.uk}

\author{Yuichi Yoshida}
\orcid{}
\affiliation{%
	\institution{Principles of Informatics Research Division\\National Institute of Informatics}
		\country{}
}
\email{yyoshida@nii.ac.jp}
	
\author{Gramoz Goranci}
\orcid{}
\affiliation{%
	\institution{School of Computing Science\\University of Glasgow}
	\country{}
}
\email{gramoz.goranci@glasgow.ac.uk}

	\begin{abstract}	
Effective resistance is an important metric that measures the similarity of two vertices in a graph. It has found applications in  graph clustering, recommendation systems and network reliability, among others. In spite of the importance of the effective resistances, we still lack efficient algorithms to exactly compute or approximate them on massive graphs.

In this work, we design several \emph{local algorithms} for estimating effective resistances, which are algorithms that only read a small portion of the input while still having provable performance guarantees. To illustrate, our main  algorithm approximates the effective resistance between any vertex pair $s,t$ with an arbitrarily small additive error $\varepsilon$ in time $O(\poly(\log n/\varepsilon))$, whenever the underlying graph has bounded mixing time. We perform an extensive empirical study on several benchmark datasets, validating the performance of our algorithms.
	\end{abstract}




\maketitle

\section{Introduction}

Metrics that capture the similarity between vertices in a graph have played a pivotal role in the quest for understanding the structure of large-scale networks. Typical examples include personalized PageRank (PPR)~\cite{page1999pagerank}, Katz similarity~\cite{katz1953new} and SimRank~\cite{jeh2002simrank}, each of which can be thought of as a random walk-based measure on graphs. These metrics have found applications in recommender systems~\cite{kim2011personalized}, link prediction~\cite{lu2011link,song2009scalable}, etc. %

A remarkably important random walk-based metric for measuring vertex similarity is the \emph{effective resistance}. Given a graph $G$ treated as a resistor network, the effective resistance $R_G(s,t)$ between two vertices $s,t$ in $G$ is the energy dissipation in the network when routing one unit of current from  $s$ to $t$. It is well known that the effective resistance is inherently related to the behaviour of random walks on graphs\footnote{We only consider simple random walks in the paper: suppose we are at vertex $v$, we jump to a neighbor of $v$ with probability $1/\deg(v)$, where $\deg(v)$ is the degree of vertex $v$.}. Concretely, the effective resistance between $s$ and $t$ is proportional to the \emph{commute time} $\kappa(s,t)$, defined as the expected number of steps a random walk starting at $s$ visits vertex $t$ and then goes back to $s$~\cite{chandra1996electrical}. Using this interpretation, we can deduce that the smaller $R_G(s,t)$ is, the more similar two vertices $s,t$ are.

Indeed, effective resistance has proven ubiquitous in numerous applications including graph clustering~\cite{alev_et_al:LIPIcs:2018:8369,fortunato2010community}, recommender systems~\cite{kunegis2007collaborative}, measuring robustness of networks~\cite{ellens2011effective}, spectral sparsification~\cite{spielman2011graph}, graph convolutional networks~\cite{AHMAD2021389}, location-based advertising~\cite{haoran2019analyzing}, among others. Moreover, in the theoretical computer science community, the use of effective resistance has led to a breakthrough line of work %
for provably speeding up the running time of many flow-based problems in combinatorial optimization~\cite{Anari2015,Christiano2011,Madry2014}.

Despite of the importance of effective resistance, we still lack efficient methods to compute or approximate them on massive graphs. For any two vertices $s,t$ and approximation parameter $\varepsilon>0$, one can $(1+\varepsilon)$-approximate $R_G(s,t)$  in\footnote{Throughout the paper, we use $\tilde{O}$ to hide polylogarithmic factors, i.e., $\tilde{O}(f(n))=O(f(n)\cdot \poly\log f(n))$.}  %
$\tilde{O}(m \log (1/\varepsilon))$ time~\cite{CohenKMPPRX14}, where $m$ denotes the number of edges in a graph. There exists an algorithm that $(1+\varepsilon)$-approximates \emph{all-pairs} effective resistances in $\tilde{O}(n^2/\varepsilon)$ time~\cite{JS18:sketch}. These results, though theoretically competitive, require access to the entire input graph. Given the rapid growth of modern networks, such polynomial time algorithms (even those running in near linear time in the number of vertices and edges) are prohibitively costly. %
This motivates the following question:
\begin{center}
	\emph{Can we obtain a competitive estimation to $R_G(s,t)$ while exploring only a small portion of the graph?}
\end{center}

We address this question by exploiting the paradigm of \emph{local} or \emph{sub-linear} algorithms. This computational model is particularly desirable in applications where one requires the effective resistances amongst only a few number of vertex pairs. %
Despite that the effective resistance is a key tool in large-scale graph analytics, designing local algorithms for estimating it is a largely unexplored topic.

In this paper, we provide several local algorithms for estimating pairwise effective resistances with provable performance guarantees. For any specified vertex pair $s,t$, our algorithms output an estimate of $R_G(s,t)$ with an arbitrarily small constant additive error, while  exploring a small portion of the graph.
To formally state our results, we utilize the well-known \emph{adjacency list} model~\cite{ron2019sublinear}, which assumes query access to the input graph $G$ and supports the following types of queries in constant time:
\begin{itemize}
	\item \emph{degree query}: for any specified vertex $v$, the algorithm can get the degree $\deg(v)$ of $v$;
	\item \emph{neighbor query}: for any specified vertex $v$ and index $i\leq \deg(v)$, the algorithm can get the $i$-th neighbor of $v$;
	\item \emph{uniform sampling}: the algorithm can sample a vertex $v$ of $G$ uniformly at random. %
\end{itemize}

Our main objective is to find a good estimate of the pairwise effective resistance $R_G(s,t)$ for any specified vertex pair $s,t$ by making as few queries as possible to the graph while achieving fast running time.

\paragraph{Our contributions.} We give a systemic study of local algorithms for estimating $s$-$t$ effective resistances for general graphs.
\begin{itemize}
	\item Theoretically, we provide three types of local algorithms for estimating effective resistances. All of them are based on random walks, but vary from their connections to effective resistances:
	(i) the first type is based on approximating the pseudo inverse of the Laplacian matrix,
	(ii) the second type is based on commute times, (iii) the third type is based on the number of spanning trees.
	\item We empirically demonstrate the competitiveness of our algorithms on popular benchmarks for graph data. In particular, for certain real-world networks, we will see that our algorithms run $10^5$ to $10^6$ faster than existing polynomial-time methods and estimate effective resistance to within a multiplicative error of $0.1$.
\end{itemize}

To illustrate, our main local algorithm approximates $R_G(s,t)$ with an arbitrarily small additive error $\varepsilon$ in time $O(\poly(\log n/\varepsilon))$, whenever the underlying graph has bounded mixing time, which is justified in real-world networks. Previously, the only work on this problem was by Andoni et al.~\cite{andoni2018solving}, and it achieves $(1+\varepsilon)$-approximation to $R_G(s,t)$ in $O(\frac{1}{\varepsilon^2}\poly\log \frac{1}{\varepsilon})$ time for $d$-regular expander graphs. Indeed, one of our algorithms for general graphs is based on   \cite{andoni2018solving}. %

Using the fact that the length of shortest paths and effective resistances are exactly the same on tree graphs, we can observe that graphs with large mixing time do not admit efficient local algorithms. Concretely, let us consider a path graph on $n$ vertices. It is known that the path graph has large mixing time, and there is no local algorithm that makes a sub-linear number of queries and approximates the length of shortest paths within a constant multiplicative factor or additive error, thus giving us the same impossibility result for effective resistances. This suggests that our bounded mixing time assumption is necessary to design local algorithms with sublinear number of queries and running time.

\section{Related work}%
In this section, we discuss some related work. %

Hayashi~et~al.~\cite{hayashi2016efficient} gave an algorithm for approximating the effective resistances of vertex pairs that are endpoints of edges.  %
Their algorithm is based on sampling spanning trees uniformly at random, and it %
$(1+\varepsilon)$-approximates $R_G(s,t)$ for \emph{every} $(s,t)\in E$ in expected running time  $\lceil\log(2m/\delta)/2\varepsilon^2\rceil\cdot\sum_{u\in V}\pi_G(u)\kappa_G(u,r)$, where $\pi_G(u)$ denotes the stationary probability at a vertex $u\in V$ of a random walk on $G$, $\kappa_G(u,v)$ denotes the commute time between two vertices $u,v\in V$ and $r\in V$ is some vertex. %

There also exist several local algorithms for other random walk based quantities, such as the stationary distribution, PageRank, Personalized PageRank and transition probabilities.

{\it The stationary distribution.} Lee et al.~\cite{lee2013computing} and Bressan~et~al.~\cite{bressan2019approximating} studied the question of computing the stationary distribution $\pi$ of a Markov Chain locally. These algorithms take as input any state $v$, and answer if the stationary probability of $v$ exceeds some $\Delta\in (0,1)$ and/or output an estimate of $\pi(v)$. They only make use of a local neighborhood of $v$ on the graph induced by the Markov chain and run in sublinear time for some families of Markov Chains.

{\it PageRank.}
Borgs~et~al.~presented a method for identifying all vertices whose PageRank is larger than some threshold~\cite{borgs2012sublinear}. Specifically, for a threshold value $\Delta\geq 1$ and a constant $c>3$, with high probability, their algorithm returns a set $S\subseteq V$ such that $S$ contains all vertices with PageRank at least $\Delta$ and no vertex with PageRank at least $\Delta/c$. The algorithm runs in $\tilde{O}(\frac{n}{\Delta})$ time.

Bressan~et~al.~developed a sub-linear time algorithm that employs local graph exploration~\cite{bressan2018sublinear}. Their algorithm $(1+\varepsilon)$-approximates the PageRank of a vertex on a directed graph. For constant $\varepsilon>0$, the algorithm runs in $\tilde{O}(\min(m^{3/4}\Delta^{1/4}d^{-3/4},m^{6/7}d^{-5/7}))$, where $\Delta$ and $d$ are respectively the maximum and average outdegree.

{\it Personalized PageRank (PPR).} The PPR $\pi_s(t)$ of a start vertex $s$ and target vertex $t$ measures the frequency of visiting $t$ via short random-walks from $s$. For a given threshold $\delta$ such that $\pi_s(t)>\delta$, Lofgren~et~al.~solved this with small relative error and an expected running time of $O(\sqrt{d/\delta})$~\cite{lofgren2014fast}, where $d$ is the average in-degree of the graph. Their algorithm is based on a bi-directional search technique and an improved implementation was presented in~\cite{lofgren2016personalized}.

{\it Transition probabilities.} Another problem related to effective resistance is estimating transition probabilities in a Markov chain. Specifically, given transition matrix $P$, initial source distribution $\sigma$, target state $t$, and a fixed length $\ell$, the goal is to estimate the probability $p$ that an $\ell$-step random walk starting from distribution $\sigma$ ends at $t$. %
Banerjee and Lofgren developed an algorithm that can estimate such a probability with respect to a minimum threshold $\delta$ such that $p>\delta$ by employing a bi-directional approach~\cite{banerjee2015fast}. Specifically, their algorithms returns an estimator $\hat{p}$ of $p$ such that with high probability $|\hat{p}-p|<\max\{\varepsilon p, \delta\}$ for any $\varepsilon>0$. %

\section{Preliminaries}
Let $G=(V,E)$ be an undirected graph. For any $v\in V$, we let $\deg(v)$ denote the degree of $v$. The \emph{volume} of a set $S$ of vertices, denoted $\vol(S)$, is the sum of their degrees. Furthermore, for a set $S\subseteq V$, the \emph{conductance} of $S$, denoted $\phi_G(S)$, is the number of edges with one endpoint in $S$ and the other in $V\setminus S$ divided by $\vol(S)$. The \emph{conductance} of $G$, denoted $\phi(G)$, is defined to be $\min_{S\subseteq V,0<\vol(S)\leq\frac{\vol(V)}{2}}\phi_G(S)$. A graph $G$ is called an \emph{expander} if $\phi(G)\geq \phi$ for some universal constant $\phi\in(0,1)$.

Let $\mat{A}$ denote its adjacency matrix and let $\mat{D}$ denote the degree diagonal matrix. Let $\mat{L}=\mat{D}-\mat{A}$ denote the \emph{Laplacian} matrix of $G$. Let $\mat{L}^\dagger$ denote the Moore-Penrose pseudo-inverse of the Laplacian of $G$. Let $\1_u \in \R^V$ denote the (row) indicator vector of vertex $u$ such that $\1_u(v)=1$ if $v=u$ and $0$ otherwise. Let $\chi_{s,t}=\1_s-\1_t$.
\begin{definition}
	Given any two vertices $u,v\in V$, the \emph{$s$-$t$ effective resistance} is defined as $$R_G(s,t):=\chi_{s,t}\mat{L}^\dagger\chi_{s,t}^\top=\mat{L}^\dagger_{s,s}-2\mat{L} ^\dagger_{s,t}+\mat{L}^\dagger_{t,t}.$$
\end{definition}

\paragraph{Random walks.} Given a graph $G$, we consider the simple random walk on $G$: suppose we are currently at $v$, then we jump to a neighbor $u$ with probability $\frac{1}{\deg(v)}$. We use $\mat{P}:=\mat{D}^{-1}\mat{A}$ to denote the random walk transition matrix. Let $\lambda=\max\{|\lambda_2|, |\lambda_n|\}$, where $\lambda_i$ is the $i$-th largest eigenvalue of the matrix $\mat{P}$.
\begin{definition}
	The commute time $\kappa(s,t)$ between vertices $s,t$ is the expected number of steps in a random walk that starts at vertex $s$ visits vertex $t$ and then comes back to $s$.
\end{definition}

Random walks on graphs are a type of Markov Chain. A Markov chain is said to be \emph{positive recurrent} if, starting in any state $i$, the expected time until the process returns to state $i$ is finite. A Markov chain is said to be \emph{aperiodic} if for any state $i$ there are no restrictions on when it is possible for the process to enter state $i$.

\begin{definition}
	A Markov chain is said to be ergodic if it is aperiodic and positive recurrent.
\end{definition}

Informally, the \emph{mixing time} of the graph $G$ refers to the number of steps needed before a random walk on $G$ converges to its stationary distribution. We refer to \cite{sinclair1992improved} for a formal definition. It is known that the spectral gap $1-\lambda$ is intimately related to the mixing time of $G$. That is, the larger $1-\lambda$ is, the smaller mixing time is, and vice versa.

\section{The Local Algorithms}\label{sec:algorithms}
\subsection{Algorithms based on approximating Laplacian inverse}
In this section, we provide local algorithms for effective resistances by approximating the Laplacian pseudo-inverse $\mat{L}^\dagger$ of the graph.

\paragraph{High-level idea} Our algorithm works for general graphs and is based on the aforementioned sublinear-time algorithm for $d$-regular graphs~\cite{andoni2018solving}. The basic idea is as follows. Recall that by definition of effective resistance, $R_G(s,t)=\chi_{s,t}\mat{L}^{\dagger}\chi_{s,t}^\top$ and $\mat{P}=\mat{D}^{-1}\mat{A}$ is the random walk transition matrix. Using the Neumann series of the matrix $\mat{L}^{\dagger}$ {(see Lemma \ref{lemma:neumann})}, we can write  %
\begin{align*}
	R_G(s,t)&=\chi_{s,t} \sum_{i=0}^\infty \mat{P}^{i}\mat{D}^{-1} \chi_{s,t}^\top\\
	&=\chi_{s,t} \sum_{i=0}^{\ell-1} \mat{P}^{i}\mat{D}^{-1} \chi_{s,t}^\top + \chi_{s,t} \sum_{i\geq \ell}^\infty \mat{P}^{i}\mat{D}^{-1} \chi_{s,t}^\top.
\end{align*}
for any $\ell>0$. For graphs with large spectral gap (i.e., expander graphs or graphs with low random walk mixing time), we can show that for any additive error $\varepsilon$, we can choose $\ell$ appropriately such that the second term is at most $\varepsilon/2$, and the first term can be approximated within additive error $\varepsilon/2$. For the latter, we use a simple Monte Carlo approach (i.e., to use the empirical distribution of the endpoints of a small number of random walks) to approximate the quantity $\1_s \mat{P}^i\1_t^\top$, the (transition) probability that a length-$i$ random walk starting from $s$ ends at $t$, for any $i\geq 1$.

Now we introduce one assumption, building upon which we present and analyze two local algorithms.
\begin{assumption}\label{assume:expander}
	Let $G$ be a connected graph with minimum vertex degree at least $1$. Further, assume that the Markov chain corresponding to the random walk on $G$ is ergodic.
\end{assumption}

\paragraph{\bf The first algorithm: \textsc{EstEff-TranProb}} We first present an algorithm that uses the above idea. Recall that $\lambda=\max\{|\lambda_2|, |\lambda_n|\}$, where $\lambda_i$ is the $i$-th largest eigenvalue of the matrix $\mat{P}$.

\begin{theorem}\label{thm:esteff-expander}
	Under Assumption~\ref{assume:expander}, there is an algorithm \Call{EstEff-TranProb}{$G,\varepsilon, s,t$} (see Algorithm~\ref{alg:tranprob}) that outputs an estimate $\hat{\delta}_{s,t}$ such that with probability at least $9/10$, it holds that
	\[
	\abs{R_G(s,t)-\hat{\delta}_{s,t}} \leq \varepsilon.
	\]
	
	The running time and query complexity of the algorithm are $O(\ell^4(\log \ell) /\varepsilon^2)$ for $\ell=\frac{\log (4/(\varepsilon-\varepsilon\lambda))}{\log (1/\lambda)}$.
\end{theorem}
The above algorithm is very efficient, if the graph has small $\lambda$, or has low mixing time, a property that is satisfied by many real networks.
Now we present the algorithm \textsc{EstEff-TranProb}.
\begin{algorithm}[h!]
	\caption{\textsc{EstEff-TranProb}$(G,\varepsilon,s,t)$}\label{alg:tranprob}
	\DontPrintSemicolon
	$\ell=\frac{\log (4/(\varepsilon-\varepsilon\lambda))}{\log (1/\lambda)}$\;
	$r\gets 40\ell^2 (\log (80\ell))/\varepsilon^2$\;
	\For{$i:=0,1,\ldots,\ell-1$}{
		\label{alg:transition1}	Perform $r$ independent random walks of length $i$ starting at $s$, and let $X_{i,s}$ (resp., $X_{i,t}$) be the number of walks that end at $s$ (resp., $t$). \;
		\label{alg:transition2}	Perform $r$ independent random walks of length $i$ starting at $t$, and let $Y_{i,s}$ (resp., $Y_{i,t}$) be the number of walks that end at $s$ (resp., $t$) \;
		Set $\hat{\delta}_{s,t}^{(i)}= \frac{X_{i,s}}{r \deg(s)} -\frac{X_{i,t}}{r \deg(t)}-\frac{Y_{i,s}}{r \deg(s)} +\frac{Y_{i,t}}{r \deg(t)}   $\;
	}
	\Return{$\hat{\delta}_{s,t}= \sum_{i=0}^{\ell-1}\hat{\delta}_{s,t}^{(i)}$}
\end{algorithm}

\textbf{Proof of Theorem~\ref{thm:esteff-expander}} We first note that the running time and query complexity of the algorithm are $O(r\ell^2)=O(\ell^4(\log \ell) /\varepsilon^2)$.

In the following, we prove the correctness of the algorithm.
We first present a basic property of effective resistance. Let $\mat{Q}=\mat{D}^{-1/2}\mat{A}\mat{D}^{-1/2}$.
Recall that $\mat{L}=\mat{D}-\mat{A}=\mat{D}^{1/2}(I-\mat{Q})\mat{D}^{1/2}$ and that $R_G(s,t)=\chi_{s,t}\mat{L}^{\dagger}\chi_{s,t}^\top$. Note that $\mat{Q}=\mat{D}^{-1/2}\mat{A}\mat{D}^{-1/2}=\mat{D}^{1/2}\mat{P}\mat{D}^{-1/2}$ is symmetric and is similar to $\mat{P}$ (as the diagonal matrix $\mat{D}$ is invertible, which in turn follows from Assumption~\ref{assume:expander}). %
We let $\lambda_1\geq \lambda_2\geq \lambda_3\ge \cdots\ge \lambda_{n}$ be the eigenvalues of $\mat{Q}$ (and also $\mat{P}$ by the similarity of $\mat{P}$ and $\mat{Q}$), with corresponding (row) orthonormal eigenvectors $w_1,w_2,w_3,\ldots, w_n$, i.e,. $w_j \mat{Q} = \lambda_j w_j$. It is known that $\lambda_1=1$ and $w_1=\frac{\1_V D^{1/2}}{\sqrt{2m}}$.

\begin{lemma}
	\label{lemma:neumann}
	It holds that
	\begin{eqnarray*}
		R_G(s,t) = \chi_{s,t} \sum_{i=0}^\infty \mat{P}^{i}\mat{D}^{-1} \chi_{s,t}^\top.
	\end{eqnarray*}
\end{lemma}
\begin{proof}
	By the spectral decomposition of $\mat{Q}$, we have that for any integer $i\geq 0$, $\mat{Q}^i=\sum_{j=1}^n\lambda_j^i w_j^\top w_j=w_1^\top w_1+\sum_{j=2}^n\lambda_j^i w_j^\top w_j$.
	
	Since $\mat{L}=\mat{D}-\mat{A}=\mat{D}^{1/2}(I-\mat{Q})\mat{D}^{1/2}$, we have that
	\begin{align*}
		\mat{L}^\dagger
		&=\mat{D}^{-1/2}{(I-\mat{Q})}^\dagger \mat{D}^{-1/2} \\
		&= \mat{D}^{-1/2}\sum_{j=2}^n\frac{1}{1-\lambda_j}w_j^\top w_j \mat{D}^{-1/2}=\mat{D}^{-1/2}\sum_{j=2}^n\sum_{i=0}^{\infty} \lambda_j^i w_j^\top w_j \mat{D}^{-1/2}\\
		&= \mat{D}^{-1/2} \sum_{i=0}^{\infty}  \sum_{j=2}^n\lambda_j^i w_j^\top w_j \mat{D}^{-1/2} = \mat{D}^{-1/2} \sum_{i=0}^{\infty}  \left( \mat{Q}^i - w_1^\top w_1 \right) \mat{D}^{-1/2}.
	\end{align*}
	
	Now we write $\chi_{s,t}\mat{D}^{-1/2} = \sum_{j=1}^n \alpha_j w_j$. We note that $\alpha_1=\chi_{s,t}\mat{D}^{-1/2} \cdot w_1^\top = \chi_{s,t} \mat{D}^{-1/2} \mat{D}^{1/2}\1_V^\top/\sqrt{2m} = 0$. Thus $\chi_{s,t}\mat{D}^{-1/2} = \sum_{j=2}^n \alpha_j w_j$. Then $\chi_{s,t}\mat{D}^{-1/2}\mat{Q}^i \mat{D}^{-1/2} \chi_{s,t}^\top= \sum_{j=1}^n \alpha_j^{2} \lambda_j^i$.
	
	Note that $\mat{P}^i \mat{D}^{-1}={(\mat{D}^{-1}\mat{A})}^i \mat{D}^{-1} =\mat{D}^{-1/2}{(\mat{D}^{-1/2}\mat{A}\mat{D}^{-1/2})}^i \mat{D}^{-1/2}=\mat{D}^{-1/2}\mat{Q}^i \mat{D}^{-1/2}$.
	Thus
	\begin{align*}
		R_G(s,t) &=\chi_{s,t}\mat{L}^\dagger \chi_{s,t}^\top=\chi_{s,t} \mat{D}^{-1/2} \sum_{i=0}^{\infty}  (\mat{Q}^i-w_1^\top w_1) \mat{D}^{-1/2} \chi_{s,t}^\top \\
		&= \sum_{i=0}^{\infty}  \chi_{s,t} \mat{D}^{-1/2} \mat{Q}^i \mat{D}^{-1/2} \chi_{s,t}^\top
		=\sum_{i=0}^\infty \chi_{s,t} \mat{P}^i \mat{D}^{-1} \chi_{s,t}^\top. \qedhere
	\end{align*}
\end{proof}

By Assumption~\ref{assume:expander}, the Markov chain corresponding to the random walk on $G$ is ergodic. Then $\lambda_2<1$. Recall that $\lambda=\max\{|\lambda_2|, |\lambda_n|\}$ { and that $\chi_{s,t}\mat{D}^{-1/2} = \sum_{j=1}^n \alpha_j w_j$, where $\alpha_1 =0$}. Furthermore, by the assumption that each vertex has degree at least $1$, we have $\sum_{j=2}^n \alpha_j^{2} =\norm{\chi_{s,t}\mat{D}^{-1/2}}_2^2\leq \norm{\chi_{s,t}}_2^2=2$. Now we prove the following two claims.

\begin{claim}
	It holds that $\abs{R_G(s,t) - \sum_{i=0}^{\ell-1} \chi_{s,t} \mat{P}^i \cdot \mat{D}^{-1}\chi_{s,t}^\top}\leq \frac{\varepsilon }{2}$.
\end{claim}
\begin{proof}
	It holds that
	\begin{align*}
		&\lrabs{R_G(s,t) - \sum_{i=0}^{\ell-1} \chi_{s,t} \mat{P}^i \cdot \mat{D}^{-1}\chi_{s,t}^\top} \\
		& = \lrabs{\sum_{i=0}^{\infty} \chi_{s,t} \mat{P}^i \cdot \mat{D}^{-1}\chi_{s,t}^\top - \sum_{i=0}^{\ell-1} \chi_{s,t} \mat{P}^i \cdot \mat{D}^{-1}\chi_{s,t}^\top}\\
		&= \lrabs{\sum_{i=\ell}^\infty \chi_{s,t}\mat{D}^{-1/2}\mat{Q}^i \mat{D}^{-1/2} \chi_{s,t}^\top}=\lrabs{\sum_{i=\ell}^\infty \sum_{j=2}^n \alpha_j^{2} \lambda_j^i} \\
		&\leq \sum_{i=\ell}^\infty\lambda^i \sum_{j=2}^n \alpha_j^{2} \leq \frac{2\lambda^\ell}{1-\lambda}  \leq \frac{\varepsilon}{2},
	\end{align*}
	where the last inequality follows from $\ell=\frac{\log (4/(\varepsilon-\varepsilon\lambda))}{\log (1/\lambda)}$.
\end{proof}

\begin{claim}
	With probability at least $9/10$,
	\[
	\lrabs{\hat{\delta}_{s,t}-\sum_{i=0}^{\ell-1} \chi_{s,t} \mat{P}^i \cdot \mat{D}^{-1}\chi_{s,t}^\top}\leq \frac{\varepsilon}{2}.
	\]
\end{claim}
\begin{proof}
	We observe that for any $i\geq 0$,
	\begin{align*}
		\chi_{s,t} \mat{P}^i \cdot \mat{D}^{-1}\chi_{s,t}^\top & = (\1_s-\1_t)P^i \cdot \mat{D}^{-1} {(\1_s-\1_t)}^\top \\
		&= \frac{\1_s P^i\1_s^\top}{\deg(s)} -  \frac{\1_s \mat{P}^i\1_t^\top}{\deg(t)} -\frac{\1_t \mat{P}^i\1_s^\top}{\deg(s)} + \frac{\1_t \mat{P}^i\1_t^\top}{\deg(t)}.
	\end{align*}
	
	Note that for any $0\leq i\leq \ell-1$, in the algorithm, we perform $r$ random walks of length $i$ from $s$. Since $X_{i,s}$ is the number of walks that end at $s$ and $\1_s \mat{P}^i\1_s^\top$ is exactly the probability of a random walk of length $i$ from $s$ ends at $s$, we have that
	\[
	\mathbb{E}X_{i,s}=r \cdot \1_s \mat{P}^i\1_s^\top.
	\]
	Furthermore, by the Chernoff-Hoeffding bound,
	\begin{align*}
		&\mathbb{P}\left[\lrabs{\frac{X_{i,s}}{r \deg(s)}-\frac{\1_s \mat{P}^i\1_s^\top}{\deg(s)}}\geq\frac{\varepsilon}{8\ell } \right]  =\mathbb{P}\left[\lrabs{\frac{X_{i,s}}{r \deg(s)}-\frac{\mathbb{E}[X_{i,s}]}{r \deg(s)}}\geq\frac{\varepsilon}{8\ell } \right]\\
		&=\mathbb{P}\left[\lrabs{{X_{i,s}}-{\mathbb{E}[X_{i,s}]}}\geq\frac{r \deg(s)\varepsilon}{8\ell } \right]\\
		&\leq 2 \exp(-2{\deg(s)}^2\varepsilon^2 r^2/(64\ell^2 r))\leq 2 \exp(-\varepsilon^2r/(32\ell^2)) \leq \frac{1}{40\ell},
	\end{align*}
	where the last inequality follows from $r=40\ell^2 (\log (80\ell))/\varepsilon^2$.
	Similarly,
	\[
	\mathbb{P}\left[\lrabs{\frac{X_{i,t}}{r\deg(t)}-\frac{\1_s \mat{P}^i\1_t^\top}{\deg(t)}}\geq\frac{\varepsilon}{8\ell } \right]\leq \frac{1}{40\ell},
	\]
	\[
	\mathbb{P}\left[\lrabs{\frac{Y_{i,s}}{r\deg(s)}-\frac{\1_t \mat{P}^i\1_s^\top}{\deg(s)}}\geq\frac{\varepsilon}{8\ell } \right]\leq \frac{1}{40\ell},
	\]
	\[
	\mathbb{P}\left[\lrabs{\frac{Y_{i,t}}{r\deg(t)}-\frac{\1_t \mat{P}^i\1_t^\top}{\deg(t)}}\geq\frac{\varepsilon}{8\ell } \right]\leq \frac{1}{40\ell}.
	\]
	Thus by a union bound, it holds that
	\begin{align*}
		&\lrabs{\hat{\delta}_{s,t}-\sum_{i=0}^{\ell-1} \chi_{s,t} \mat{P}^i \cdot \mat{D}^{-1}\chi_{s,t}^\top}\\
		&=\Biggl|\sum_{i=0}^{\ell-1}\left( \frac{X_{i,s}}{r \deg(s)} -\frac{X_{i,t}}{r \deg(t)}-\frac{Y_{i,s}}{r \deg(s)} +\frac{Y_{i,t}}{r \deg(t)}   \right) \\
		& - \sum_{i=0}^{\ell-1}\left( \frac{\1_s \mat{P}^i\1_s^\top}{\deg(s)} -  \frac{\1_s \mat{P}^i\1_t^\top}{\deg(t)} -\frac{\1_t \mat{P}^i\1_s^\top}{\deg(s)} + \frac{\1_t \mat{P}^i\1_t^\top}{\deg(t)}  \right)\Biggr|\\
		&\leq \frac{\varepsilon}{8\ell}\cdot \ell \cdot 4=\frac{\varepsilon}{2}
	\end{align*}
	with probability $1-4\cdot \ell \cdot \frac{1}{40\ell}=\frac{9}{10}$.
\end{proof}
Therefore, with probability at least $9/10$, it holds that
\[
\abs{R_G(s,t)-\hat{\delta}_{s,t}} \leq \varepsilon.
\]
This finishes the proof of Theorem~\ref{thm:esteff-expander}.

\paragraph{\bf The second algorithm: \textsc{EstEff-TranProb-Collision}} In the previous algorithm, we used the simple Monte Carlo approach to approximate the transition probabilities (which correspond to Line 4 and 5 in Algorithm~\ref{alg:tranprob}). %
Now we give a more efficient procedure to estimate the transition probability $\1_s \mat{P}^i\1_t^\top$. Such an algorithm is based on the idea of treating the term $\1_s \mat{P}^i\1_t^\top$ (roughly) as a collision probability of two random walks of length $i/2$, starting from $s$ and $t$, respectively. In particular, if $p=\1_s \mat{P}^i\1_t^\top$, then for typical vertices $s,t$, we can approximates the probability $p$ in $O(1/\sqrt{p})$ time, in contrast to the $O(1/p)$ time from the Monte Carlo approach. This idea of approximating transition probability has been given in~\cite{banerjee2015fast}. We use this idea to present a new algorithm whose performance guarantee is given in the following theorem.
\begin{theorem}\label{thm:expander-II}
	Suppose that Assumption~\ref{assume:expander} holds. Suppose further that for any $i\leq \ell$,
	\[
	\norm{\1_s \mat{P}^i \mat{D}^{-1/2}}_2^2, \norm{\1_t \mat{P}^i \mat{D}^{-1/2}}_2^2\leq \beta_i,
	\]
	for some parameters $\beta_i$'s.
	The Algorithm~\ref{alg:expander2} (i.e., \Call{EstEff-TranProb-Collision}{$G,s,t$}) outputs an estimate $\hat{\delta}_{s,t}$ such that with probability at least $9/10$, it holds that
	\[
	\abs{R_G(s,t)-\hat{\delta}_{s,t}} \leq \varepsilon.
	\]
	
	The running time and query complexity of the algorithm are $O(\sum_{i=0}^{\ell-1}r_i)={O(\frac{\ell^{3/2}}{\varepsilon}\sum_{i=0}^{\ell-1}\sqrt{\beta_i}+\frac{\ell^3}{\varepsilon^2}\sum_{i=0}^{\ell-1}\beta_i^{3/2})}$.
\end{theorem}

\paragraph{On the choice of $\beta_i$:} Note that the algorithm is parametrized by $\beta_i$'s. We note that for expander graphs or graphs with low mixing time, it holds that $\beta_i$ is a number that exponentially decreases in terms of $i$, i.e,. $\beta_i \leq c^i$ for some constant $c<1$, as long as $\ell$ is not too large. The reason is that in an expander graph $G$ with $\phi(G)\geq \phi$ for some constant $\phi$, it holds that $||\1_s \mat{P}^i \mat{D}^{-1/2}||^2_2\leq\frac{1}{\vol(V_G)}+(1-\frac{\phi^2}{4})^{2i}$ for any starting vertex $s$ (see e.g.,~\cite{Chung:1997}). Therefore, in this case, the running time in Theorem~\ref{thm:expander-II} will be dominated by $O(\frac{\ell^3}{\varepsilon^2})$, which is faster than Algorithm~\ref{alg:tranprob}.    %

\begin{algorithm}[h!]
	\caption{\textsc{EstEff-TranProb-Collision}$(G,\varepsilon,s,t)$}\label{alg:expander2}
	\DontPrintSemicolon
	$\ell \gets \frac{\log (4/\varepsilon(1-\lambda))}{\log (1/\lambda)}$\;
	\For{$i:=0,1,\ldots,\ell-1$}{
		{$r_i\gets 20000(\sqrt{\frac{\ell^3\beta_i}{\varepsilon^2}}+\frac{\ell^3\beta_i^{3/2}}{\varepsilon^2})$}\;
		Perform $r_i$ independent random walks of length $i_1:=\lceil i/2\rceil$ starting at $s$ (resp., $t$), and let $\overrightarrow{X}_{s}\in \mathbb{R}^V$ (resp., $\overrightarrow{X}_{t}\in \mathbb{R}^V$) be a
		\emph{row vector} whose $v$'th component is the fraction of random walks from $s$ (resp., $t$) that end up at $v$, {divided by $\sqrt{\deg(v)}$} \;
		Perform $r$ independent random walks of length $i_2:=\lfloor i/2 \rfloor$ starting at $s$ (resp., $t$), and let $\overrightarrow{Y}_{s}\in \mathbb{R}^V$ (resp., $\overrightarrow{Y}_{t}\in \mathbb{R}^V$) be a
		\emph{row vector} whose $v$'th component is the fraction of random walks from $s$ (resp., $t$) that end up at $v$, {divided by $\sqrt{\deg(v)}$} \;
		Set $\hat{\delta}_{s,t}^{(i)}= {\overrightarrow{X}_{s}\cdot \overrightarrow{Y}_{s}^\top -\overrightarrow{X}_{s}\cdot \overrightarrow{Y}_{t}^\top-\overrightarrow{X}_{t}\cdot \overrightarrow{Y}_{s}^\top +\overrightarrow{X}_{t}\cdot \overrightarrow{Y}_{t}^\top}$\;
	}
	\Return{$\hat{\delta}_{s,t}= \sum_{i=0}^{\ell-1}\hat{\delta}_{s,t}^{(i)}$}
\end{algorithm}

\begin{proof}[\bf Proof of Theorem~\ref{thm:expander-II}]
	W.l.o.g.\ we consider the case that the length $i$ of the random walk is even. Note that for any $s,t$,
	\begin{align*}
		\1_s \mat{P}^i\1_t^\top&=\1_s {(\mat{D}^{-1}\mat{A})}^i\1_t^\top =\1_s {(\mat{D}^{-1}\mat{A})}^{i/2} \mat{D}^{-1}{\left({(\mat{D}^{-1}\mat{A})}^\top\right)}^{i/2}\mat{D}\1_t^\top\\
		&=\1_s \mat{P}^{i/2} \mat{D}^{-1}{(\mat{P}^\top)}^{i/2}\mat{D}\1_t^\top=\1_s \mat{P}^{i/2} \mat{D}^{-1}{(\mat{P}^\top)}^{i/2}\mat{D}\1_t^\top\\
		&=\langle \1_s \mat{P}^{i/2} \mat{D}^{-1/2}, \1_t \mat{D} \mat{P}^{i/2}\mat{D}^{-1/2}\rangle\\
		&=\deg(t)\cdot \langle \1_s \mat{P}^{i/2} \mat{D}^{-1/2}, \1_t  \mat{P}^{i/2}\mat{D}^{-1/2}\rangle.
	\end{align*}
	Thus,
	\[
	\frac{\1_s \mat{P}^i\1_t^\top}{\deg(t)}=\langle \1_s \mat{P}^{i/2} \mat{D}^{-1/2}, \1_t  \mat{P}^{i/2}\mat{D}^{-1/2}\rangle.
	\]
	
	Note that for any vertex $v$, the quantity  $[\1_s \mat{P}^{i/2}\mat{D}^{-1}](v)$ is the probability of a length-$(i/2)$ random walk that starts from $s$ and ends at vertex $v$, divided by $\sqrt{\deg(v)}$; and the quantity $[\1_t \mat{P}^{i/2}](v)$ is the probability of a length-$(i/2)$ random walk that starts from $s$ and ends at vertex $v$, divided by $\sqrt{\deg(v)}$.
	
	Now we use the argument in the proof of Lemma 19 in~\cite{clusterability}. Specifically, let $Z_{s,t}=\overrightarrow{X}_{s}\cdot \overrightarrow{Y}_{t}^\top$, where $\overrightarrow{X}_{s}\cdot \overrightarrow{Y}_{t}^\top$ are defined in Algorithm~\ref{alg:expander2}. Then $\mathbb{E}(Z_{s,t})=(\mat{D}^{-1/2}\mat{P}^i\1_a)^\top(\mat{D}^{-1/2}P^i\1_a)$. By Chebyshev's inequality and Lemma 19 in~\cite{clusterability}, we get $\mathbb{P}[|Z_{s,t}-\mathbb{E}(Z_{s,t})|>\frac{\varepsilon}{8\ell}]<(\frac{8\ell}{\varepsilon})^2(\frac{\beta_i}{r_i^2}+\frac{2\beta_i^{3/2}}{r_i})\leq \frac{1}{40\ell}$, as we have chosen  $r_i=20000(\sqrt{\frac{\ell^3\beta_i}{\varepsilon^2}}+\frac{\ell^3\beta_i^{3/2}}{\varepsilon^2})$ in the algorithm. Then the statement of the theorem follows by analogous argument from the proof of Theorem~\ref{thm:esteff-expander}.
\end{proof}

Finally, we remark that the success probabilities of both algorithms  \textsc{EstEff-TranProb} and  \textsc{EstEff-TranProb-Collision} can be boosted to $1-\frac{1}{\poly(n)}$ by standard median trick, i.e., repeatedly run these algorithms $O(\log n)$ times and output the median. On graphs with bounded mixing time, which correspond to graphs such that $1-\lambda\geq \frac{1}{\poly(\log n)}$, the algorithms run in $O(\poly(\log n/\varepsilon))$ time.

\subsection{Algorithms based on commute times of random walks}
In this section, we provide two algorithms based on the well known connections between effective resistances and commute time/visiting probability in random walks. Let $\gamma>0$ be a threshold parameter. 
\paragraph{\bf The first algorithm: \textsc{EstEff-MC}}
We can use the commute time $\kappa(s,t)$ to approximate $R_G(s,t)$. We make use of the following results.
\begin{lemma}[\cite{nash1959random,chandra1996electrical}]\label{lemma:eff_commute}
	It holds that $\kappa(s,t)=2m R_G(s,t)$.
\end{lemma}

\begin{lemma}[Proposition 2.3 in~\cite{lovasz1993random}]\label{lemma:commute}
	The probability that a random walk starting at $s$ visits $t$ before returning to $s$ is $1/(\kappa(s,t)\cdot \pi(s))$, where $\pi(s)=\frac{\deg(s)}{2m}$ is the stationary probability of $s$.
\end{lemma}
We obtain the following corollary by Lemmas~\ref{lemma:eff_commute} and~\ref{lemma:commute}.
\begin{corollary}\label{cor:monte_carlo}
	The probability $p(s,t)$ that a random walk starting at $s$ visits $t$ before returning to $s$ is $\frac{1}{R_G(s,t)\cdot \deg(s)}$. In particular, if $R_G(s,t)\leq \gamma$, then $p(s,t)\geq \frac{1}{\gamma \cdot \deg(s)}$.
\end{corollary}

The corollary above suggests the Monte Carlo algorithm below. The algorithm performs a number of random walks, starting at vertex $s$. Then it essentially count how many times the random walk traverses from $s$ to $t$ and back.

\begin{algorithm}[h!]
	\caption{\textsc{EstEff-MC}$(G,s,t,\gamma,\varepsilon)$ %
	}		\label{alg:esteff-algo-II}
	\DontPrintSemicolon
	W.l.o.g.\ suppose that $\deg(s)\leq \deg(t)$\;
	$N_0\gets \frac{3\ln 6 \cdot \gamma \cdot \deg(s)}{\varepsilon ^2}$, $X\gets 0$\;
	\For{$i=1,\dots, N_0$}{
		Perform a random walk from $s$, and stop the walk
		\begin{enumerate}
			\item\label{cond:good} if the walk has visited $t$ and then returns to $s$. %
			\item or if the walk has return to $s$ before visiting $t$.
		\end{enumerate}
		If the item (\ref{cond:good}) occurs, $X\gets X+1$\;
	}
	\Return{$\frac{N_0}{\deg(s)\cdot X}$}
\end{algorithm}

\begin{theorem}\label{thm:algII}
	Assume that that $R_G(s,t)\leq \gamma$. Let $N_0=\frac{3\ln 6 \cdot \gamma \cdot \deg(s)}{\varepsilon ^2}$. Then with probability 2/3, Algorithm~\ref{alg:esteff-algo-II} (i.e., \textsc{EstEff-MC}) returns an $(1+\varepsilon )$-approximation for $R_G(s,t)$. The running time of the algorithm is $O(\frac{m\cdot \deg(s) \cdot \gamma^2}{\varepsilon^2})$.
\end{theorem}

We remark that the above algorithm runs in sublinear time if $\gamma=o_n(1)$, i.e., $R_G(s,t)$ is small enough. In other words, when the two vertices $s,t$ are ``similar'' enough, our algorithm will be fast.
\begin{proof}[\bf Proof of Theorem~\ref{thm:algII}]
	In Algorithm~\ref{alg:esteff-algo-II}, let $X_i$ be the indicator variable that denotes the $i$-th random walk to be successful (where we do not abort the walk because of its length). Then $\mathbb{P}(X_i=1)=p(s,t)$ where $p(s,t)$ is as defined in Corollary~\ref{cor:monte_carlo}.
	Furthermore, let $X=\sum_{i=1}^{N_0}X_i$. Observe that $\mathbb{E}(X)=N_0\cdot p(s,t)=\frac{N_0}{R_\mathrm{{eff}}(s,t)\cdot \deg(s)}$, where we have used that $p(s,t)=\frac{1}{R_\mathrm{{eff}}(s,t)\cdot \deg(s)}$. %
	
	Next, assume that $R_\mathrm{{eff}}(s,t)\leq \gamma$; let $N_0=\frac{\ln(1/\delta)\cdot 3\cdot \gamma \cdot \deg(s)}{\varepsilon ^2}\geq \frac{\ln(1/\delta)\cdot 3\cdot R_\mathrm{{eff}}(s,t) \cdot \deg(s)}{\varepsilon ^2}$,  where $\delta>0$ is a parameter that will be specified later. Using Chernoff and union bounds we find that
	$\mathbb{P}[|X-\mathbb{E}(X)|>\varepsilon '\cdot \mathbb{E}(X)]<2\cdot e^{-\frac{\varepsilon '^2\cdot N_0}{3\cdot R_\mathrm{{eff}}(s,t)\cdot \deg(u)}}\leq 2\cdot \delta$ for any $\varepsilon ' > 0$.
	Thus, we find that with probability at least $1-2\delta$, $(1-2\varepsilon ')R_\mathrm{{eff}}(s,t)\leq \frac{N_0}{\deg(u)\cdot X}\leq (1+2\varepsilon ')R_\mathrm{{eff}}(s,t)$. Now, choosing $\varepsilon' = \varepsilon /2$ yields the desired approximation ratio.
	
	As a second step, we will show that each of the random walks in Algorithm~\ref{alg:esteff-algo-II} is expected to terminate within at most $2m\gamma$ steps. Consider the two cases in which the walks terminates. Let $\gamma_i, i\in \{1,2\}$ denote the number of steps taken in the random walk in some iteration, such that $i=1$ if the first termination criterion of the loop is fulfilled and $i=2$ in the other case. Then clearly, the number of steps taken in a random walk is $\min\{\gamma_1,\gamma_2\}$. Furthermore, it holds that $\min\{\gamma_1,\gamma_2\}\leq\gamma_1$. Note that $\mathbb{E}(\gamma_1)$ is the commute time $\kappa (s,t)$. Then we find that $\mathbb{E}(\min\{\gamma_1,\gamma_2\})\leq\mathbb{E}(\gamma_1)=\kappa (s,t)$.

	Finally, let $\delta =1/3$. Then we find that Algorithm~\ref{alg:esteff-algo-II}
	\begin{itemize}
		\item runs in expected time $R_\mathrm{{eff}}(s,t)\cdot N_0\in O(\frac{m\cdot \deg(s) \cdot \gamma^2}{\varepsilon^2})$ and
		\item with probability at least $1-\delta=2/3$, $\frac{N_0}{\deg(u)\cdot X}$ is an $(1+\varepsilon )$-approximation of $R_\mathrm{{eff}}(s,t)$.
	\end{itemize}{}
	This concludes the proof.
\end{proof}

\paragraph{\bf The second algorithm: \textsc{EstEff-MC2}}
For the special case that there is an edge $(s,t)$ between the two specified vertices $s,t$, we can also make use of the following probabilistic interpretation of effective resistance.
\begin{lemma}[\cite{Shayan_recentadvances}]
	Consider an edge $(s,t)$. Then $R_G(s,t)$ is the probability that a random walk from $s$ visits $t$ for the first time using $(s,t)$.
\end{lemma}
This suggests the following Monte Carlo algorithm.

\begin{algorithm}[h!]
	\caption{\textsc{EstEff-MC2}$(G,s,t,\varepsilon,\gamma,\delta)$ %
	}\label{alg:esteff-algo-III}
	\DontPrintSemicolon
	W.l.o.g.\ suppose that $\deg(s)\leq \deg(t)$\;
	$M_0\gets\frac{\ln(1/\delta)\cdot 3}{\varepsilon ^2\cdot \gamma}$, $X\gets 0$\;
	\For{$i=1,\dots, M_0$}{
		Perform a random walk from $s$, and stop the walk
		\begin{enumerate}
			\item\label{cond:good-2} if the walk visits $t$ for the first time using the edge $(s,t)$ %
			\item or if the walk visits $t$ for the first time using any other edge.
		\end{enumerate}
		If the item (\ref{cond:good-2}) occurs, $X\gets X+1$\;
	}
	\Return{$\frac{X}{M_0}$}
\end{algorithm}

\begin{theorem}
	\label{thm:esteff-algo-III}
	For $R_G(s,t)>\gamma$, Algorithm~\ref{alg:esteff-algo-III} (i.e., \textsc{EstEff-MC2}) returns with probability (1-$\delta$) a $(1+\varepsilon)$-approximation of $R_G(s,t)$.
\end{theorem}

The proof of the above theorem is deferred to Appendix \ref{sec:app}. Note that in contrast to Algorithm~\ref{alg:esteff-algo-II}, the random walks in Algorithm~\ref{alg:esteff-algo-III} stop as soon as we have reached the destination vertex $t$. Hence, one can expect that Algorithm~\ref{alg:esteff-algo-III} runs faster than Algorithm~\ref{alg:esteff-algo-II}. Experimental comparisons of the running times can be found in Section~\ref{sec:runtime}.

\subsection{An algorithm based on estimating the number of spanning trees}\label{sec:spanningtree}
Now we present a local algorithm based on a connection to the number of spanning trees of a graph. Let $T(G)$ denote the number of spanning trees of $G$.
\begin{lemma}[Corollary 4.2 in~\cite{lovasz1993random}]\label{lemma:eff_spantree}
	Let $G$ be a graph and $s,t\in V$. Let $G'$ be the graph obtained by identifying $s$ and $t$. Then
	$$R_G(s,t)=\frac{T(G')}{T(G)}$$
\end{lemma}

Lyons and Oveis Gharan gave a local algorithm for estimating the number of spanning trees~\cite{lyons2017sharp}. %

\begin{lemma}[Corollary 1.2 in~\cite{lyons2017sharp}]\label{lemma:local_spantree}
	Let $G=(V,E)$ be a graph. In the adjacent list model, together with knowledge of $n$ and $|E|$, there exists a randomized algorithm that for any given $\varepsilon,\delta>0$, outputs an estimate $Z$ that approximates $\frac{\log T(G)}{|V|}$ within an additive error of $\varepsilon$, i.e.,
$	\lrabs{\frac{\log T(G)}{|V|} - Z}\leq \varepsilon$
	with probability at least $1-\delta$, by using only $\tilde{O}(\varepsilon^{-5}+\varepsilon^{-2}\log^2n)\log{\delta^{-1}}$ number of queries.
\end{lemma}

This suggests the following algorithm based on estimating the number of spanning trees. As remarked in~\cite{lyons2017sharp}, the assumption of having the knowledge of $|E|$ in Algorithm~\ref{alg:esteff-algo-Ib} might not even be necessary.
\begin{algorithm}
	\caption{\textsc{AppNumST}$(G,\varepsilon, \delta)$ [Algorithm 2 in~\cite{lyons2017sharp}]}\label{alg:esteff-algo-I}
	\DontPrintSemicolon
	$r \gets \lceil 90^3 {\varepsilon^{-3}}\rceil$\;
	$s \gets \sum_{1\leq t<2r} 1/t$\;
	$N \gets \lceil \frac{8\log (4/\delta)s^2}{\varepsilon ^2}\rceil $\;
	\For{$i=1 \gets N$}{
		Let x be a randomly chosen vertex of G.\;
		Sample $1\leq t<2r$ with probability $1/st$.\;
		Run a $t$-step lazy simple random walk from $x$, and let $Y_i\gets \mathbb{I}[X_t=x]$\;
	}
	Sample $\lceil 256 \log(1/\delta )(\log n)^2/\varepsilon^2\rceil$ random vertices of $G$, and let $\tilde{W}$ be the average of the logarithm of twice the degree of sampled vertices.\;
	\Return{$Z:=-n^{-1}\log(4|E|)+\tilde{W}-s(\sum_{i=1}^N Y_i)/N+s/n$}
\end{algorithm}

\begin{algorithm}
	\caption{\textsc{EstEff-SpanTree}$(G,\varepsilon, \delta, u, v)$}\label{alg:esteff-algo-Ib}
	\DontPrintSemicolon
	$a \gets$ \Call{AppNumST}{$G,\frac{\varepsilon}{2}, \frac{\delta}{2}$}\;
	$b \gets$ \Call{AppNumST}{$G,\frac{\varepsilon}{2}, \frac{\delta}{2}$}\;
	\Return{$\frac{e^{a(n-1)}}{e^{bn}}$}
\end{algorithm}

\begin{theorem}\label{alg:spanningtreelocal}
	Algorithm~\ref{alg:esteff-algo-Ib} returns with probability at least $1-\delta$ an estimator $X$ such that \[e^{-\varepsilon n} R_G(s,t)\leq X\leq e^{\varepsilon n} R_G(s,t).\] The algorithm uses $\tilde{O}(\varepsilon^{-5}+\varepsilon^{-2}\log^2n)\log{\delta^{-1}}$ queries.
\end{theorem}

We give the proof of Theorem \ref{alg:spanningtreelocal} in Appendix \ref{sec:app} and remark that the above algorithm seems of theoretical interest only, as it does not perform well in practice.

\section{Experiments}
In this section, we show our experimental results.
The experiments were conducted on a Linux server with Intel Xeon E5-2643 (3.4GHz) and 768GB of main memory, and all the programs were implemented in C++ and compiled with g++ 4.8.4.
The graphs used in the experiments are taken from SNAP\footnote{\url{https://snap.stanford.edu}} and basic information about the graphs is given in Table~\ref{tab:datasets}.
We generated query pairs by randomly sampling edges $1,000$ times with replacements.

\begin{table}[h!]
	\caption{Datasets}\label{tab:datasets}
	\begin{tabular}{lrr}
		\toprule
		& $n$ & $m$ \\
		\midrule
		Facebook & 4,039 & 88,233 \\
		DBLP & 317,080 & 1,049,869 \\
		YouTube & 1,134,891 & 2,987,627 \\
		\bottomrule
	\end{tabular}
\end{table}

\begin{figure*}[t!]
	\subfloat[Facebook]{\includegraphics[width=.33\hsize]{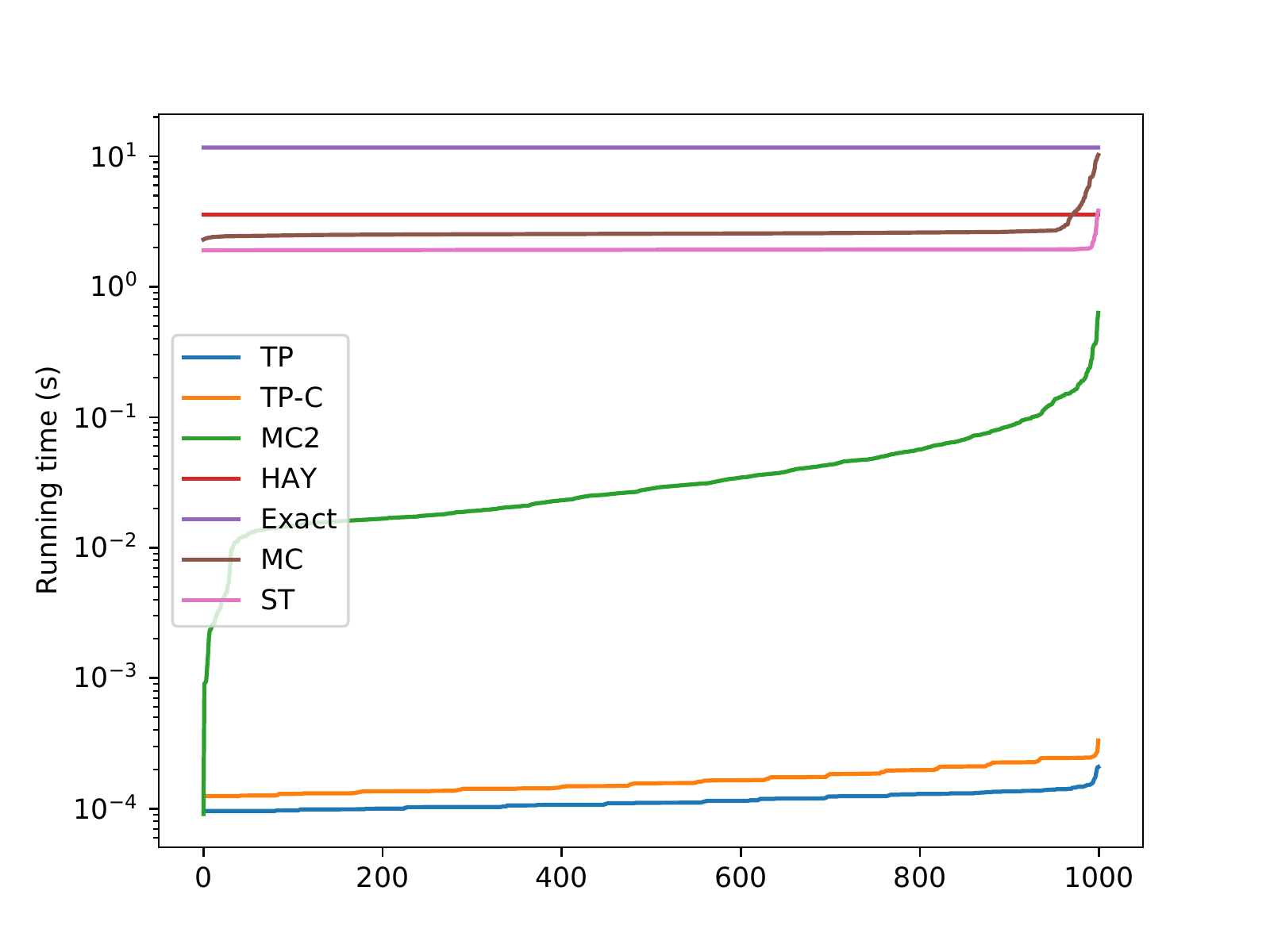}}
	\subfloat[DBLP]{\includegraphics[width=.33\hsize]{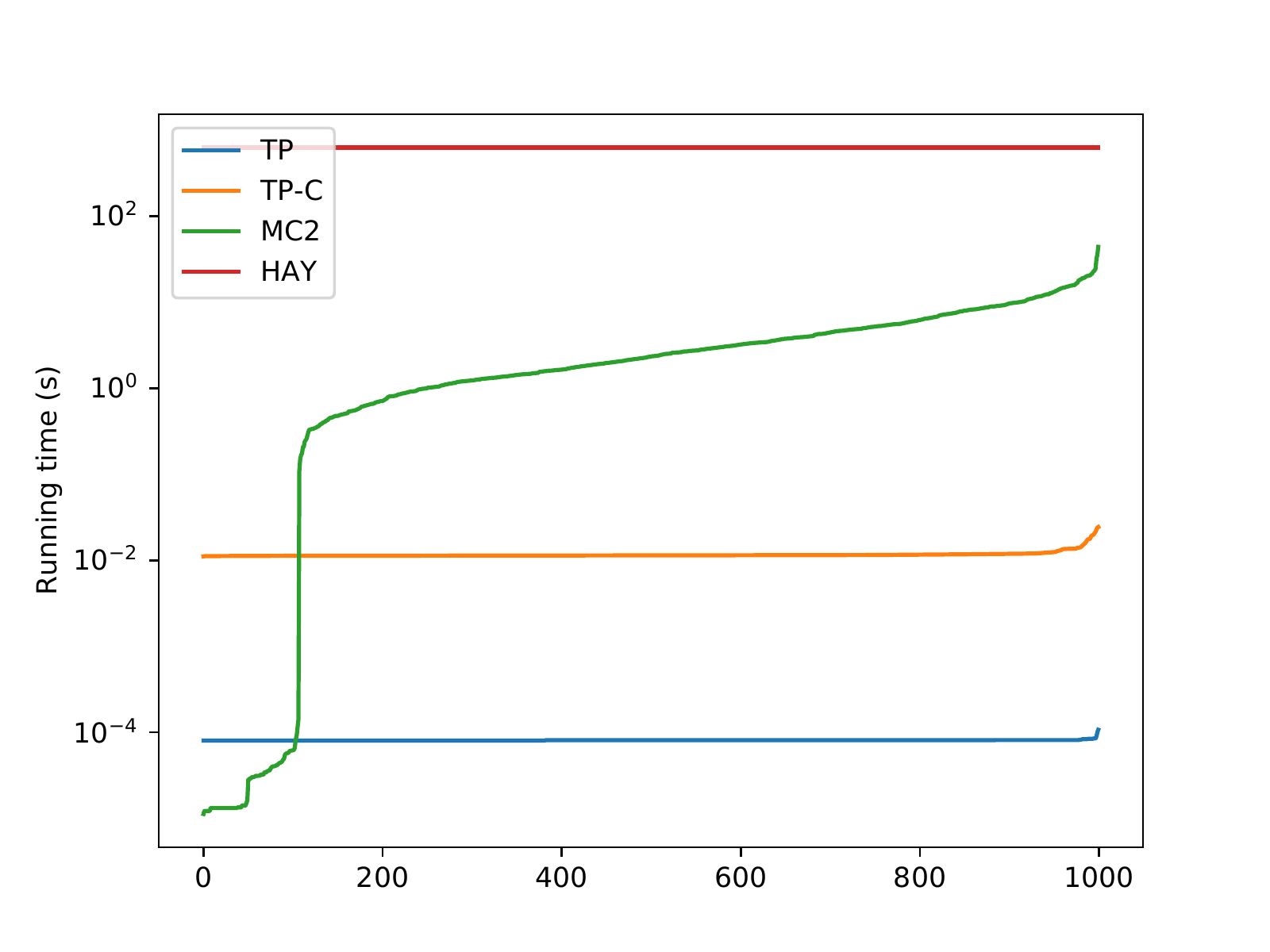}}
	\subfloat[YouTube]{\includegraphics[width=.33\hsize]{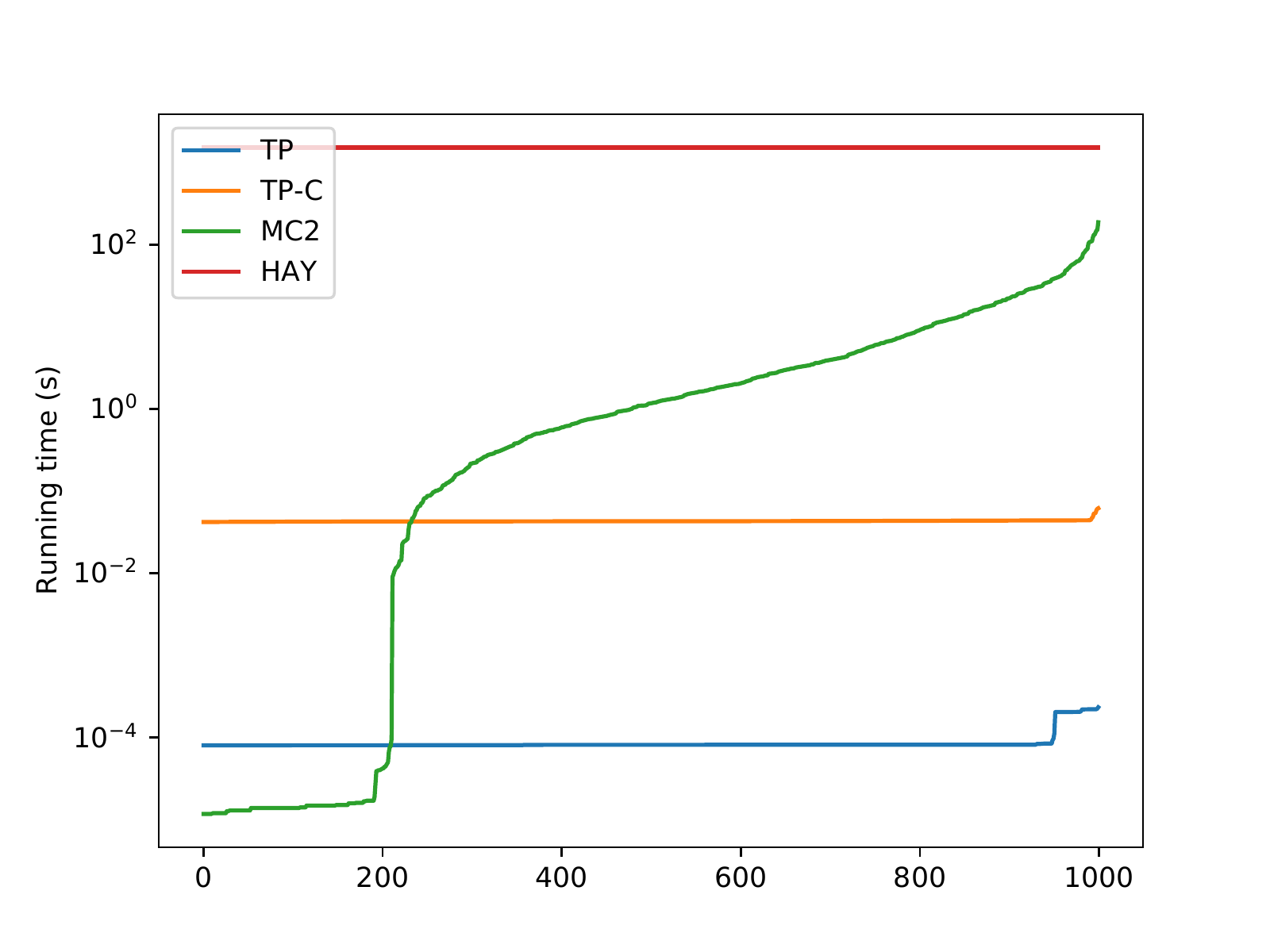}}
	\caption{Running time}\label{fig:query-time}
\end{figure*}

\begin{figure*}[t!]
	\subfloat[Facebook]{\includegraphics[width=.33\hsize]{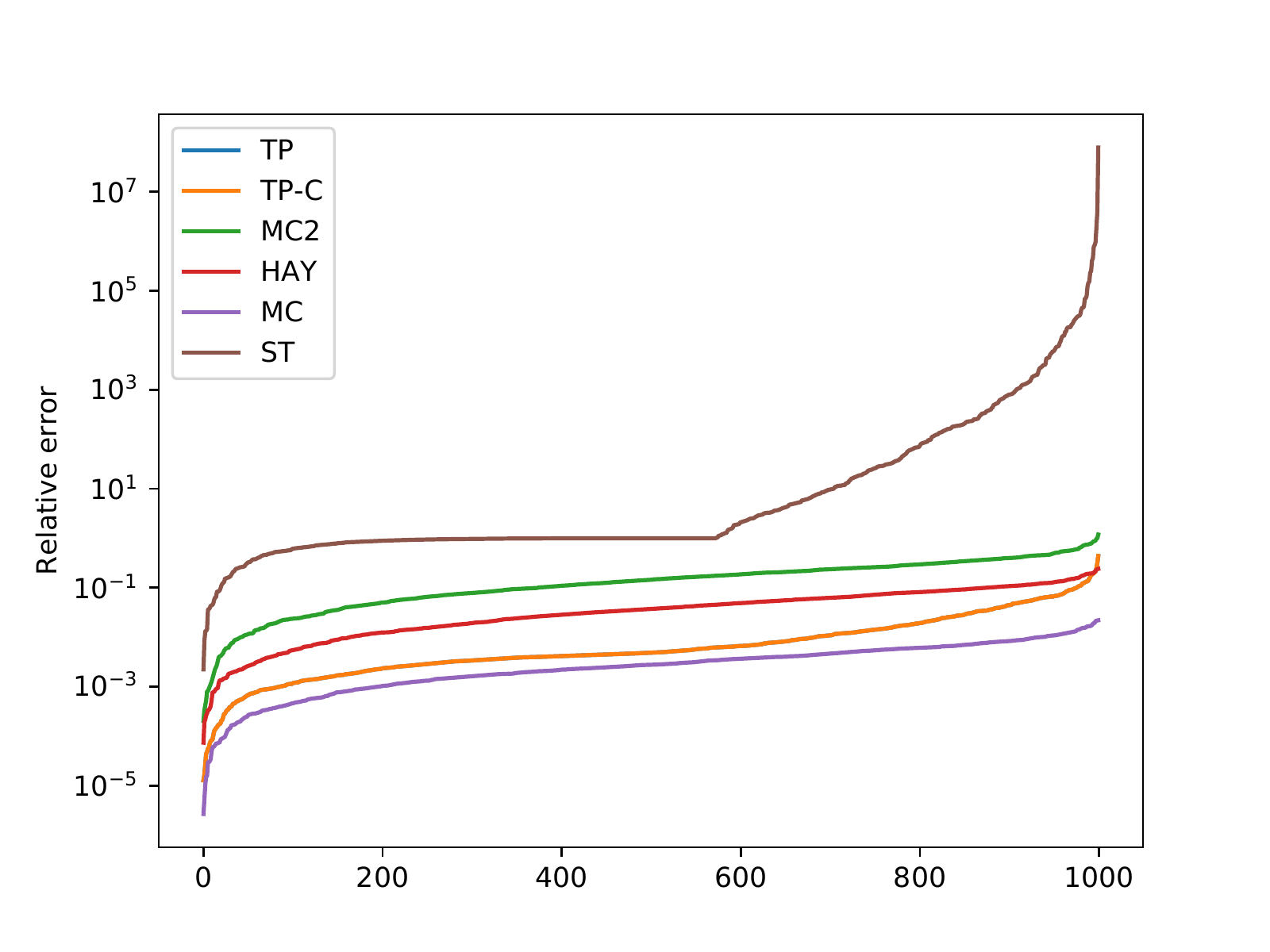}}
	\subfloat[DBLP]{\includegraphics[width=.33\hsize]{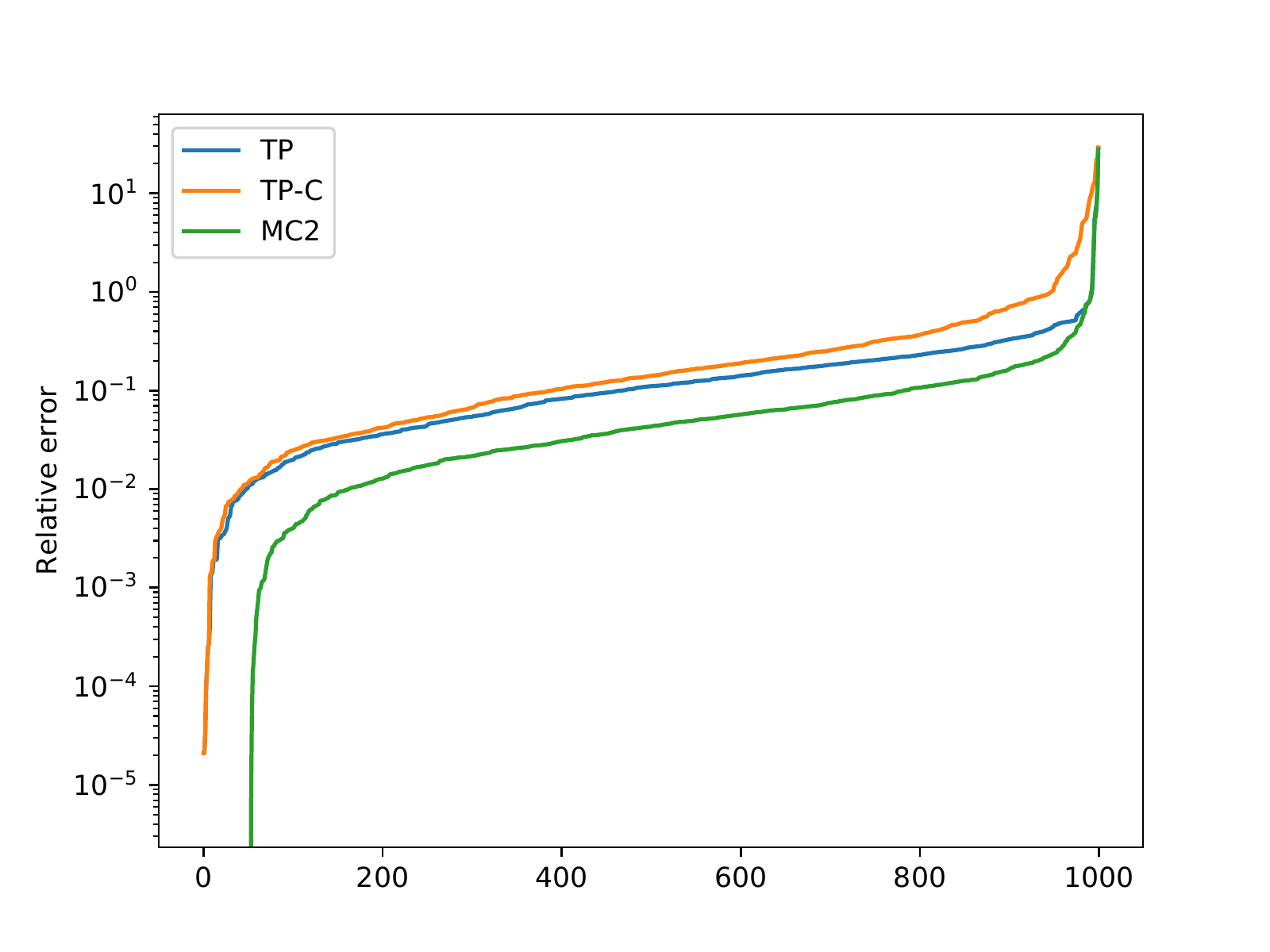}}
	\subfloat[YouTube]{\includegraphics[width=.33\hsize]{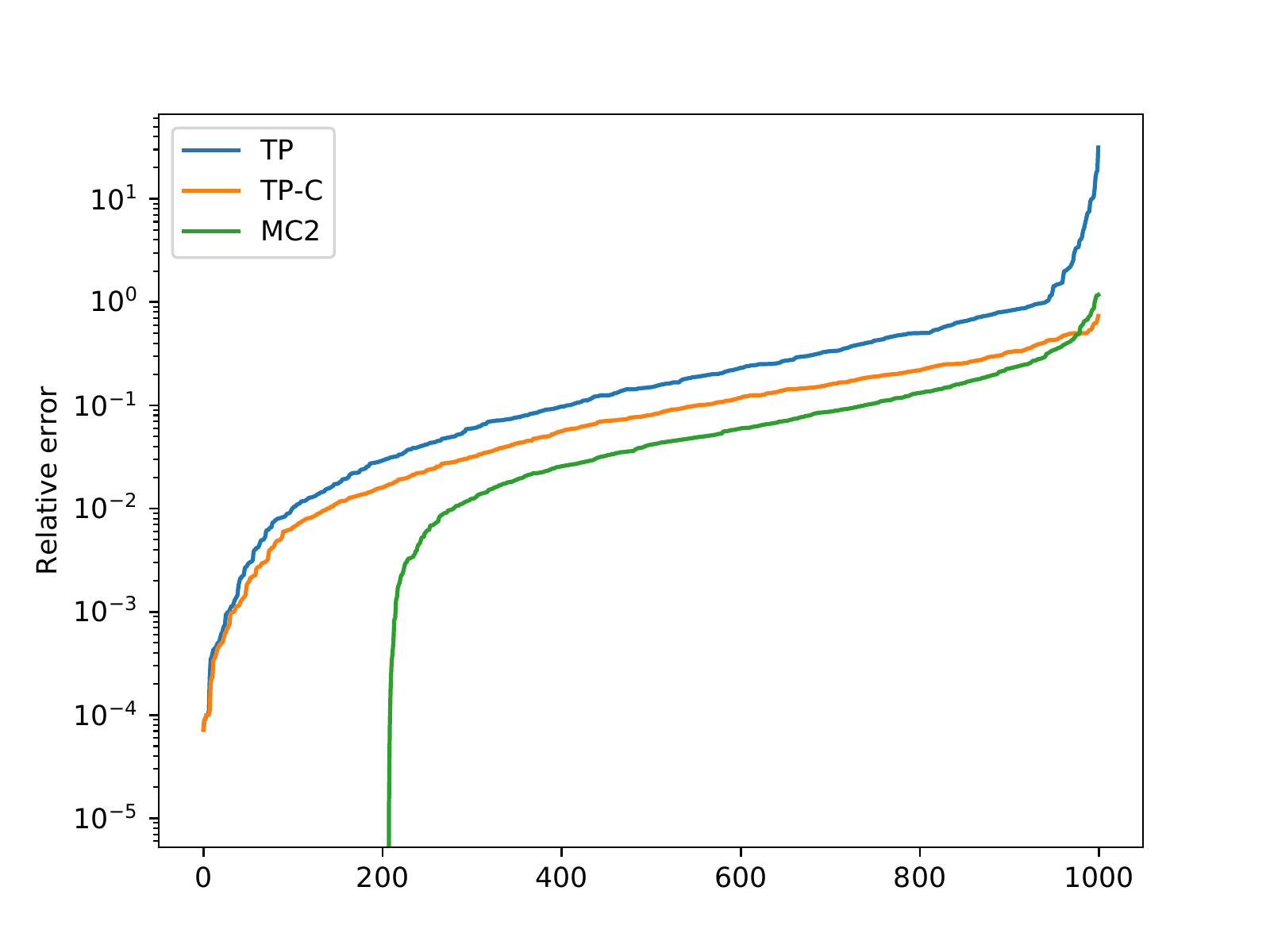}}
	\caption{Relative error}\label{fig:accuracy}
\end{figure*}

\begin{figure*}[t!]
	\begin{tabular}{lccc}
		& Facebook & DBLP & YouTube \\
		\Call{TP}{} & \includegraphics[align=c,width=.3\hsize]{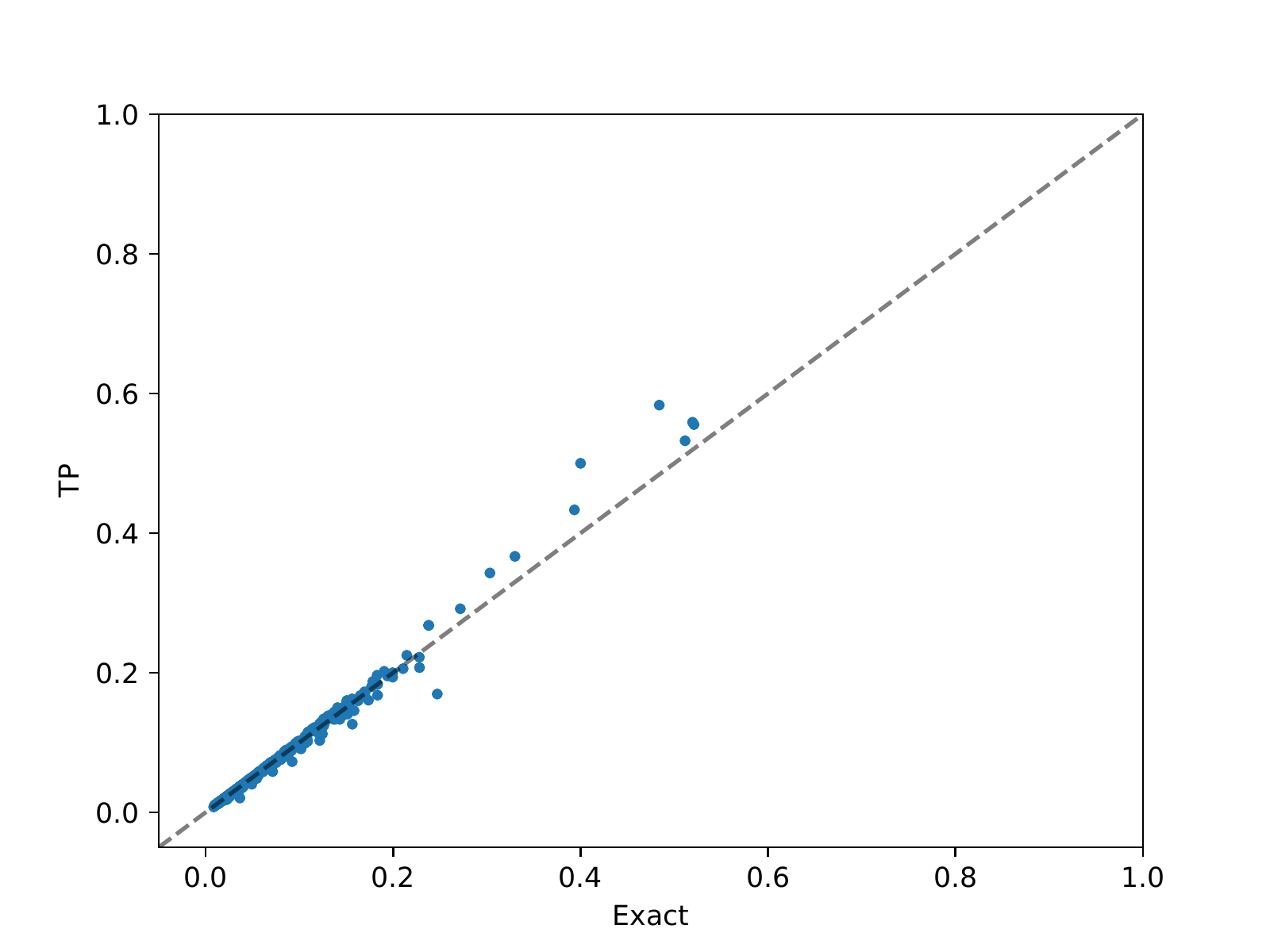} & \includegraphics[align=c,width=.3\hsize]{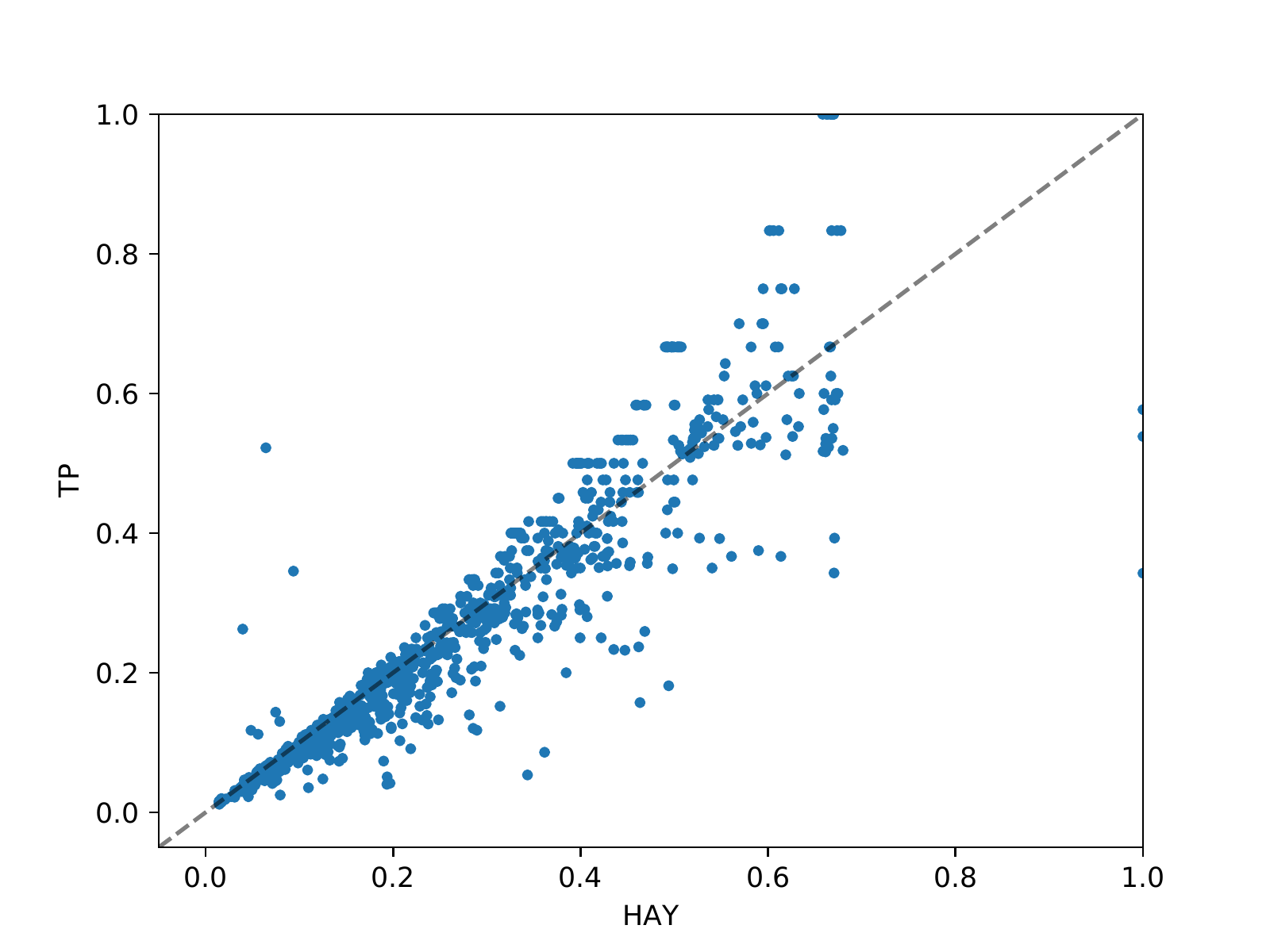} & \includegraphics[align=c,width=.3\hsize]{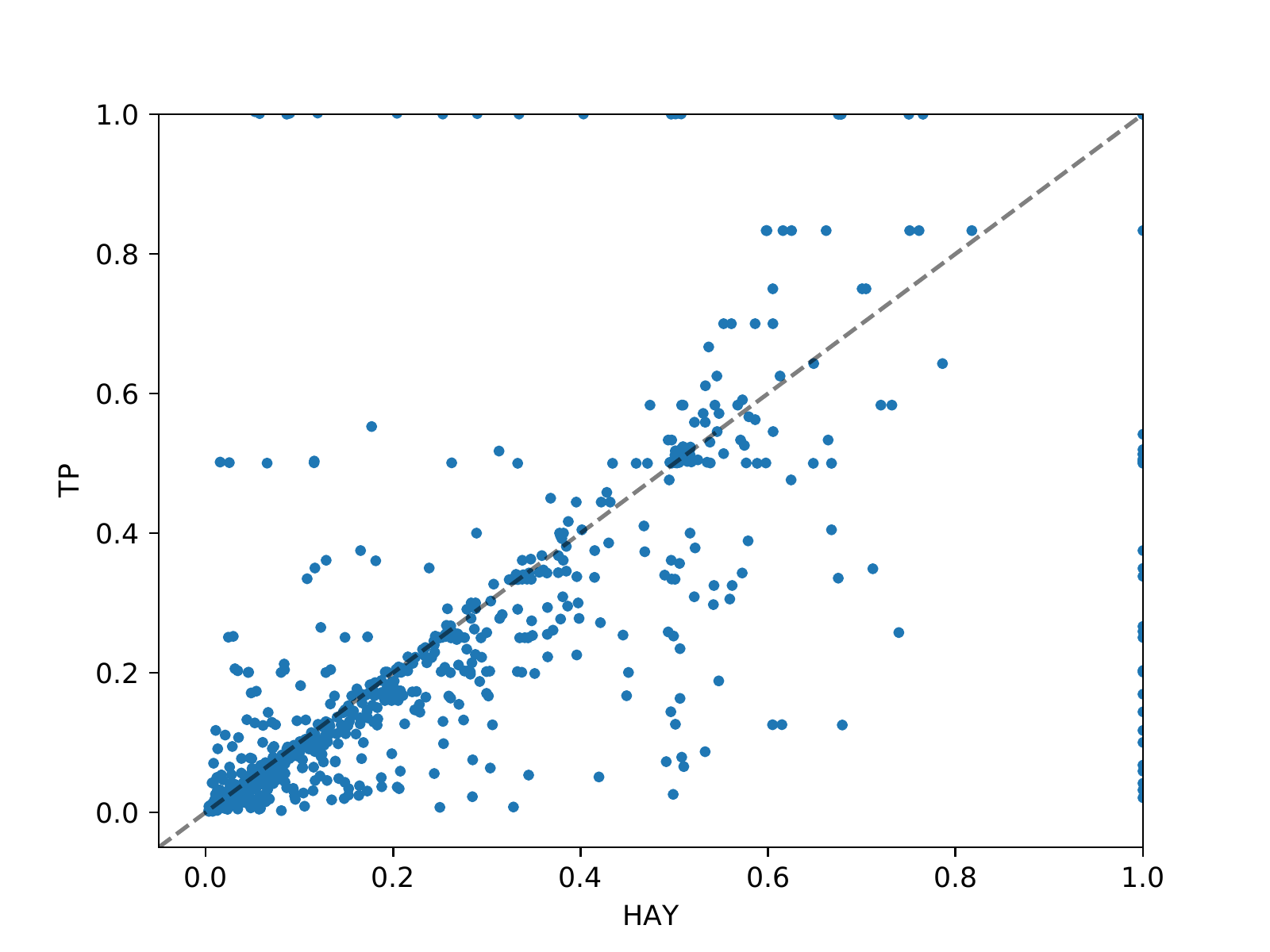} \\
		\Call{TP-C}{} & \includegraphics[align=c,width=.3\hsize]{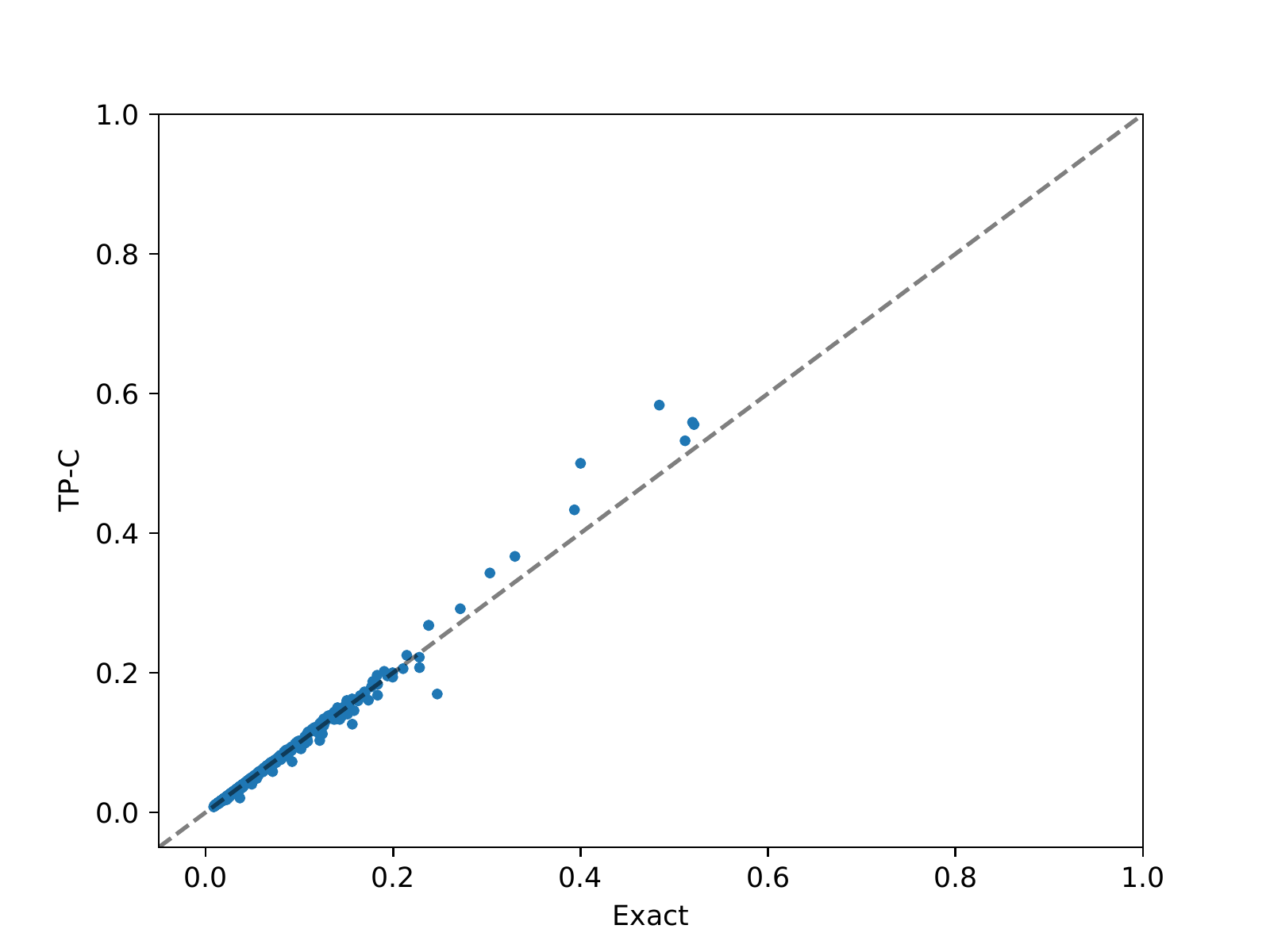} & \includegraphics[align=c,width=.3\hsize]{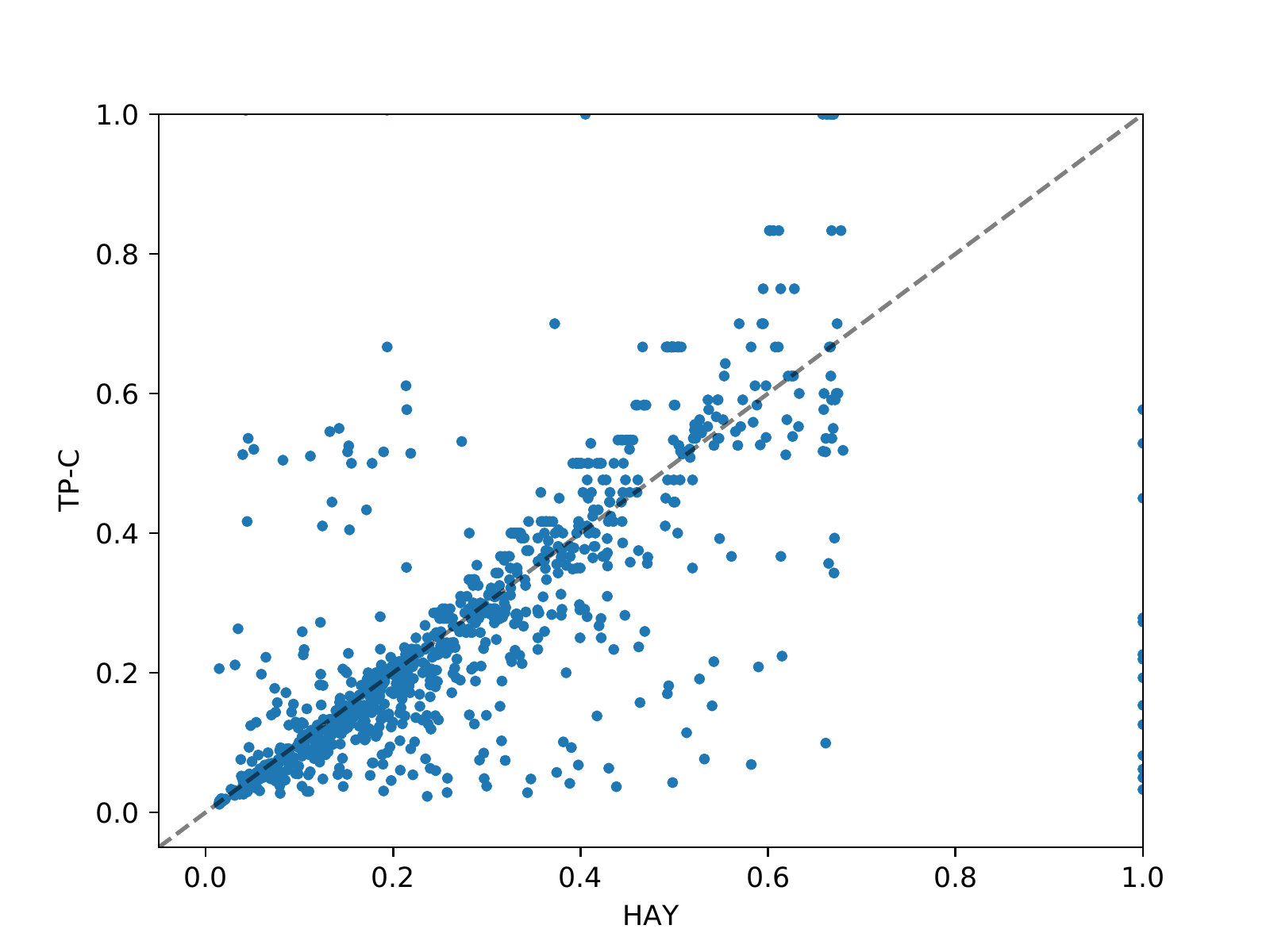} & \includegraphics[align=c,width=.3\hsize]{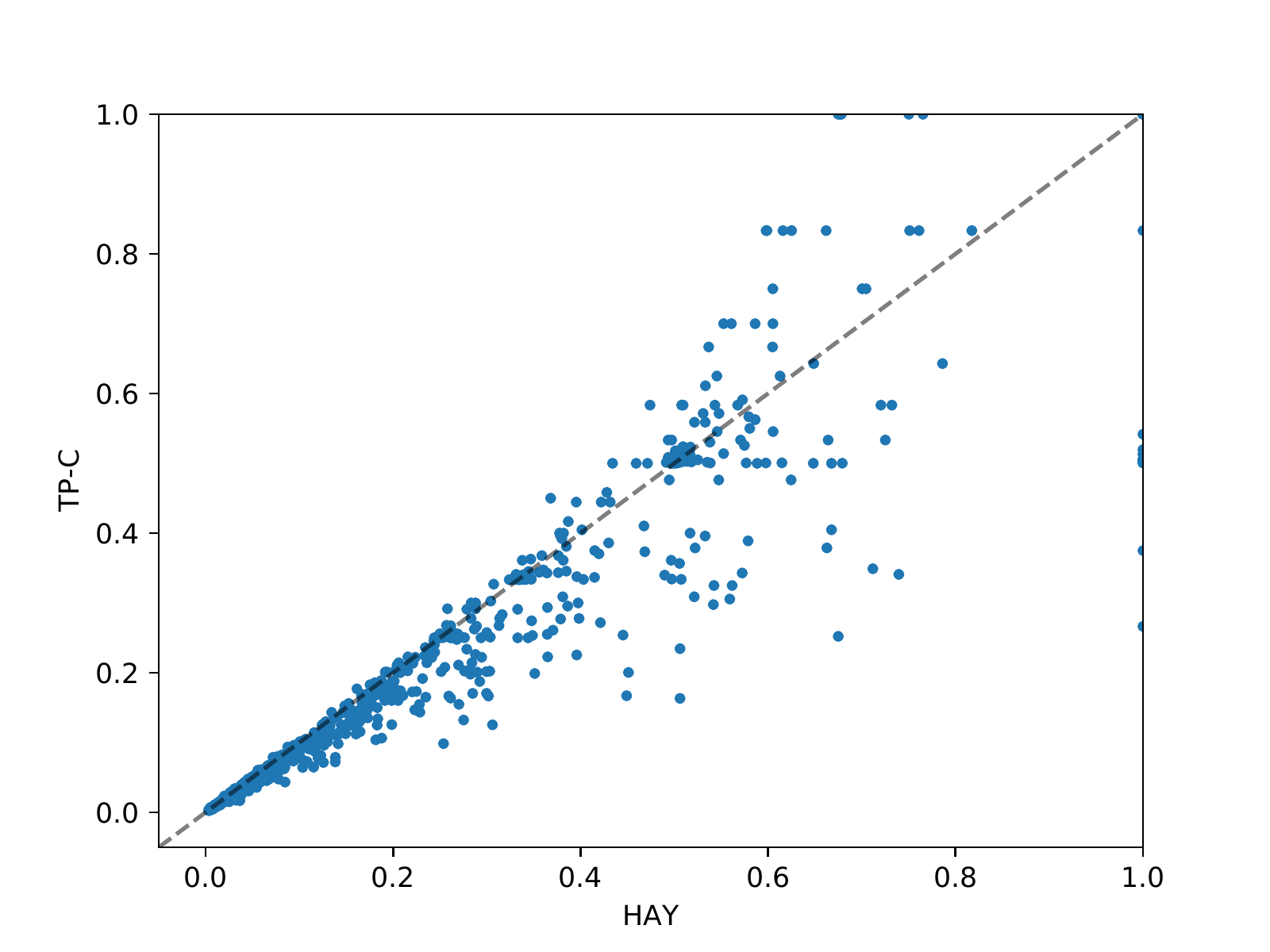} \\
		\Call{MC2}{} & \includegraphics[align=c,width=.3\hsize]{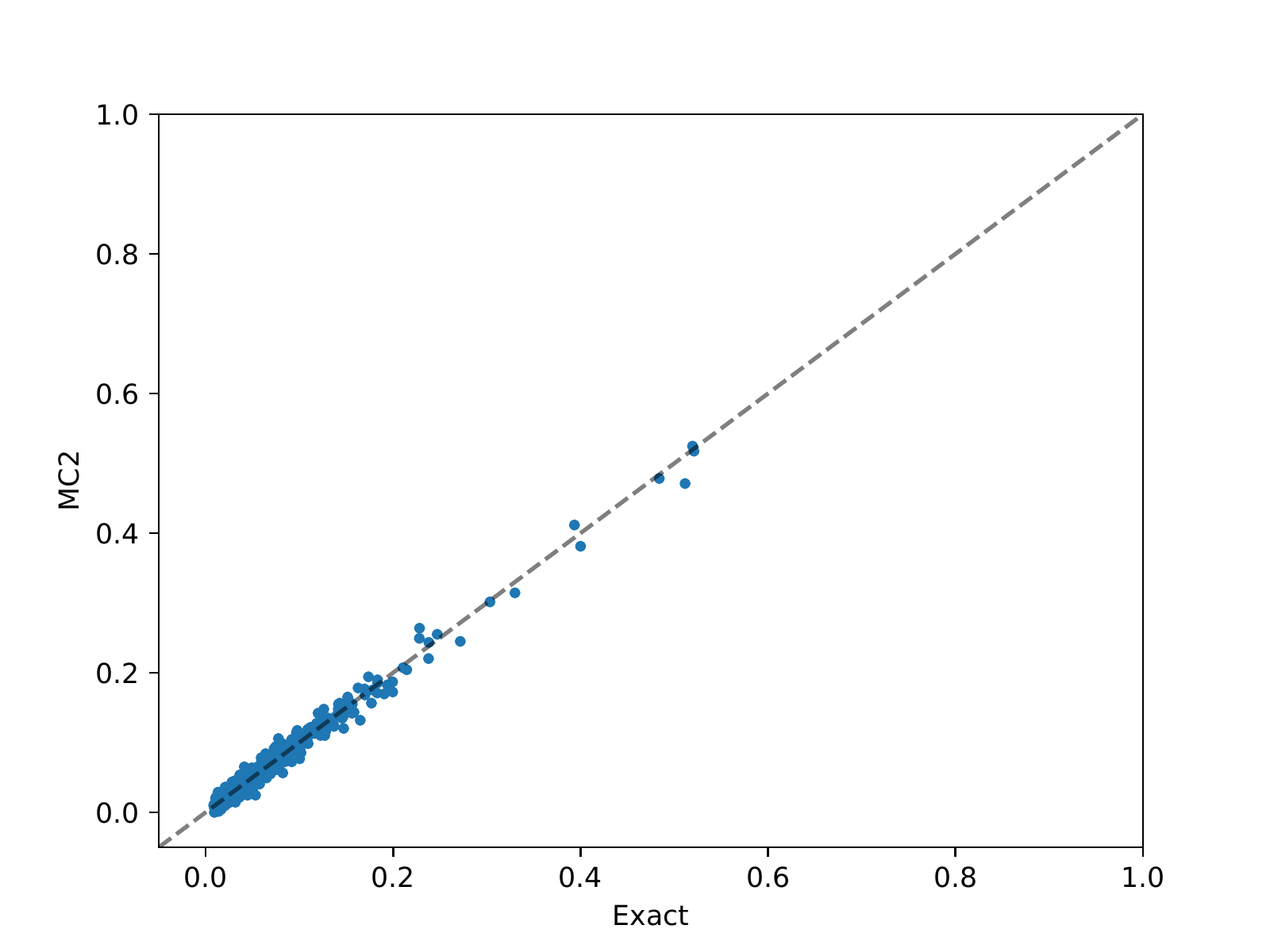} & \includegraphics[align=c,width=.3\hsize]{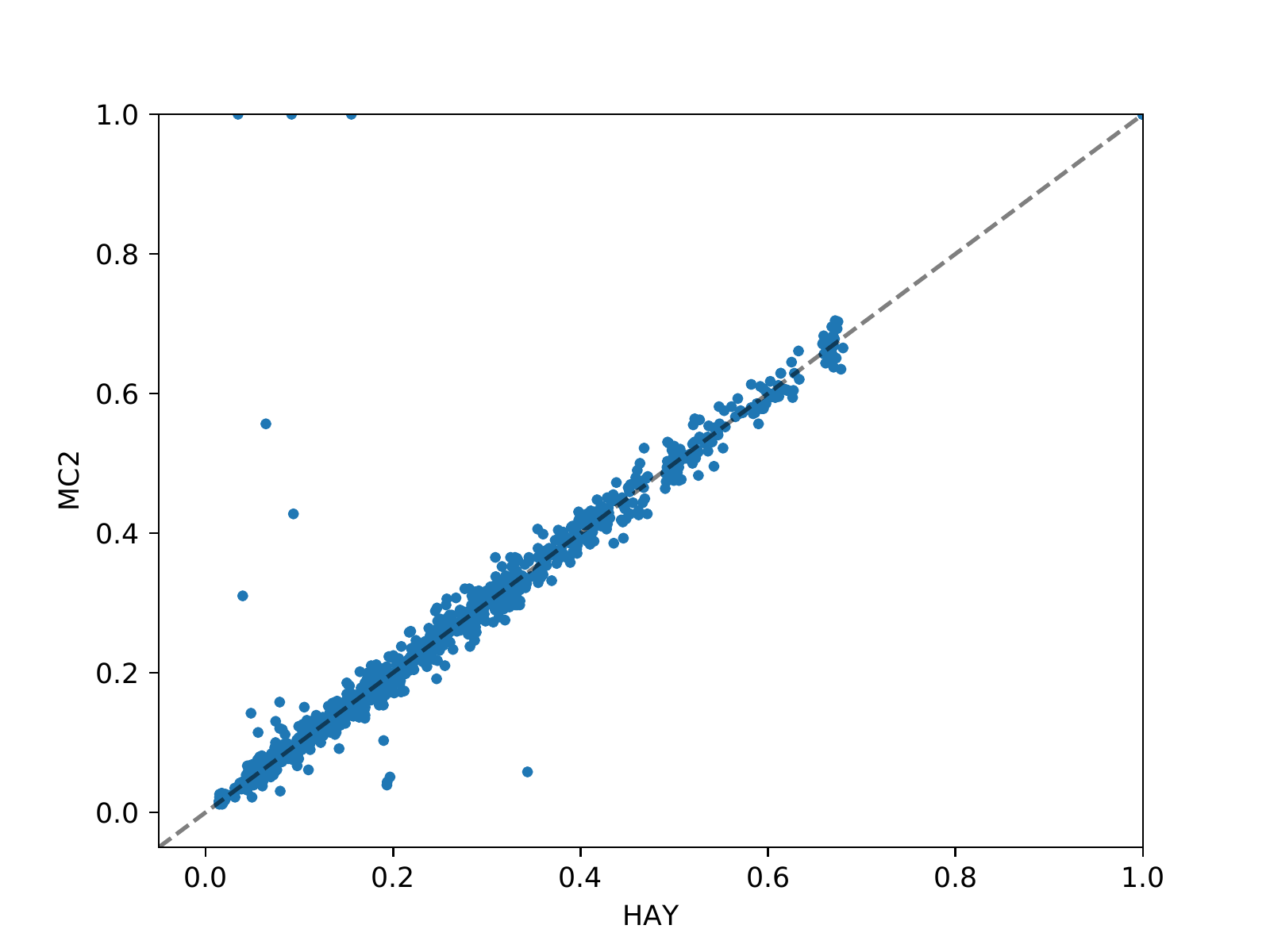} & \includegraphics[align=c,width=.3\hsize]{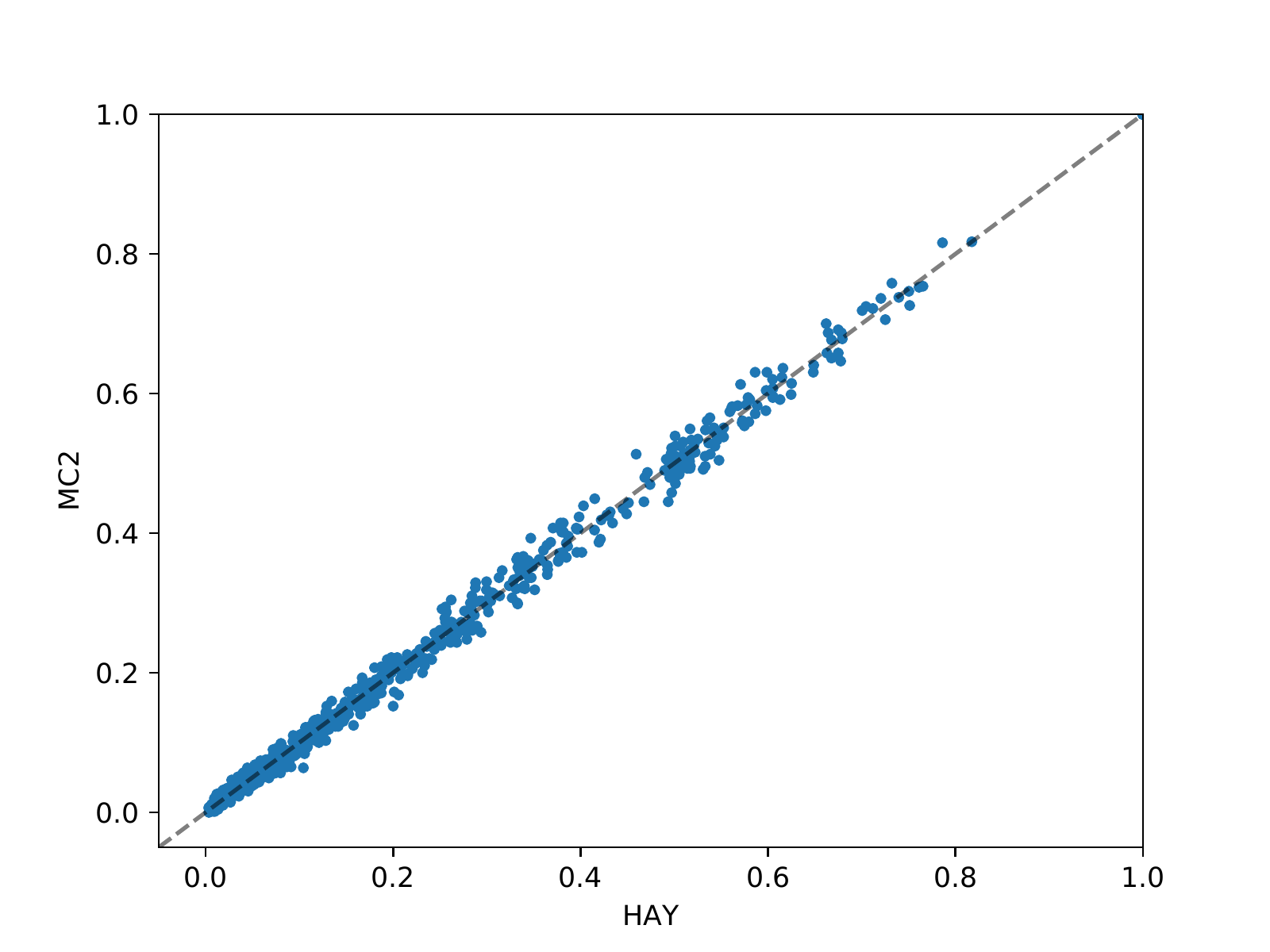} \\
	\end{tabular}
	\caption{Quality of estimated effective resistance}\label{fig:scatter}
\end{figure*}

We implemented the following algorithms:
\begin{itemize}
	\item{\Call{Exact}{}:} This method first applies the QR decomposition to the Laplacian as preprocessing and computes effective resistance according to its definition, i.e., $R_G(s,t):= \chi_{s,t}\mat{L}^\dagger\chi_{s,t}^\top$.
	\item{\Call{HAY}{}~\cite{hayashi2016efficient}:} This method computes effective resistances of all the edges at once by sampling spanning trees and it is still. %
	state-of-the-art for this problem. We fixed the number of sampled spanning trees to 10,000.
	\item{\Call{TP}{}}: Implementation of Algorithm~\ref{alg:tranprob}. We set $\varepsilon=\lambda=0.1$.
	\item{\Call{TP-C}{}}: Implementation of Algorithm~\ref{alg:expander2}. We set $\varepsilon=\lambda=0.1$.
	\item{\Call{MC}{}}: Implementation of Algorithm~\ref{alg:esteff-algo-II}. We set $\varepsilon=\gamma=0.1$.
	\item{\Call{MC2}{}}: Implementation of Algorithm~\ref{alg:esteff-algo-III}. We set $\varepsilon=\gamma=0.1$.
	\item{\Call{ST}{}}: Implementation of Algorithm~\ref{alg:esteff-algo-Ib}. We set $\varepsilon=0.1$.
\end{itemize}
To implement each algorithm, we used the same number of random walks as the one given in the corresponding pseudocode.

\subsection{Running time}\label{sec:runtime}
Figure~\ref{fig:query-time} shows the running time of each method.
For our methods, we plotted the running time for the 1,000 queries in increasing order.
For \Call{Exact}{} and \Call{HAY}{}, we plotted their preprocessing time.
We do not show the running time of \Call{Exact}{} on DBLP and YouTube because it did not terminate in 8 hours.

We can first observe that \Call{MC}{} and \Call{ST}{} on Facebook are as slow as previous (polynomial-time) algorithms, and hence we do not consider those algorithms for other graphs.

We can observe that \Call{TP}{}, \Call{TP-C}{}, and \Call{MC2}{} are much faster than the existing methods.
Note that the running time of \Call{MC2}{} depends on the queried edge $(s,t)$ because it runs until the random walk starting at $s$ reaches $t$.
In contrast, the running time of \Call{TP}{} and \Call{TP-C}{} is almost independent of the queried edge, which is preferable.
A reason that \Call{TP-C}{} is slower than \Call{TP}{} is that we need to compute inner products in \Call{TP-C}{} (Line~6 of Algorithm~\ref{alg:expander2}).

\subsection{Accuracy}
Figure~\ref{fig:accuracy} shows the accuracy of existing and our methods.
For each method, we computed the relative error as $|R - \tilde{R} |/ R$ for each query, where $R$ is the exact effective resistance for Facebook using \Call{Exact}{} and the one estimated by \Call{HAY}{} for DBLP and YouTube, and $\tilde{R}$ is the estimated effective resistance.
Then, we plotted the $1,000$ relative errors after sorting them in increasing order.
Except for ST, the relative error of our methods are within $0.1$ for most of the queries, as expected from the choice $\varepsilon = 0.1$.
Also, the results for Facebook justifies the use of \Call{HAY}{} on DBLP and YouTube as the baseline method. {In Figure \ref{fig:accuracy}(a), the results of \Call{TP}{} and \Call{TP-C}{} are very close such that their lines overlap.} The fact that the lines change concavity twice appears to be a universal phenomenon for probabilistic distributions. %

In Figure~\ref{fig:scatter}, each blue point represents $(R,\tilde{R})$ for a query, where $R$ and $\tilde{R}$ are as we defined in the paragraph above.
\Call{MC2}{} shows the best accuracy on DBLP and YouTube.
Accuracy of \Call{TP-C}{} is comparable to that of \Call{TP}{}.
Recalling that \Call{TP}{} runs faster than \Call{TP-C}{}, we can conclude that \Call{TP}{} is superior to \Call{TP-C}{}.

\section{Conclusion}
In this paper, we developed a number of local algorithms for estimating the pairwise effective resistances, a fundamental metric for measuring the similarity of vertices in a graph. Our algorithms explore only a small portion of the graph while provides a good approximation to $R_G(s,t)$ for any specified $s,t$.  Our algorithms are desirable in applications where the effective resistances of a small number of vertex pairs are needed. Our experiments on benchmark datasets validate the performance of these local algorithms.

\section*{Acknowledgments}
Y.Y. is supported in part by JSPS KAKENHI Grant Number 17H04676, 18H05291, and 20H05965.

\bibliographystyle{ACM-Reference-Format}
\bibliography{subeff}


\begin{thebibliography}{37}

%%% the Web address.

\ifx \showCODEN    \undefined \def \showCODEN     #1{\unskip}     \fi
\ifx \showDOI      \undefined \def \showDOI       #1{#1}\fi
\ifx \showISBNx    \undefined \def \showISBNx     #1{\unskip}     \fi
\ifx \showISBNxiii \undefined \def \showISBNxiii  #1{\unskip}     \fi
\ifx \showISSN     \undefined \def \showISSN      #1{\unskip}     \fi
\ifx \showLCCN     \undefined \def \showLCCN      #1{\unskip}     \fi
\ifx \shownote     \undefined \def \shownote      #1{#1}          \fi
\ifx \showarticletitle \undefined \def \showarticletitle #1{#1}   \fi
\ifx \showURL      \undefined \def \showURL       {\relax}        \fi
% invisible to TeX
\providecommand\bibfield[2]{#2}
\providecommand\bibinfo[2]{#2}
\providecommand\natexlab[1]{#1}
\providecommand\showeprint[2][]{arXiv:#2}

\bibitem[\protect\citeauthoryear{Ahmad, Jin, Lin, and Tang}{Ahmad
  et~al\mbox{.}}{2021}]%
        {AHMAD2021389}
\bibfield{author}{\bibinfo{person}{Tasweer Ahmad}, \bibinfo{person}{Lianwen
  Jin}, \bibinfo{person}{Luojun Lin}, {and} \bibinfo{person}{GuoZhi Tang}.}
  \bibinfo{year}{2021}\natexlab{}.
\newblock \showarticletitle{Skeleton-based action recognition using sparse
  spatio-temporal GCN with edge effective resistance}.
\newblock \bibinfo{journal}{\emph{Neurocomputing}}  \bibinfo{volume}{423}
  (\bibinfo{year}{2021}), \bibinfo{pages}{389 -- 398}.
\newblock
\showISSN{0925-2312}


\bibitem[\protect\citeauthoryear{Alev, Anari, Lau, and Gharan}{Alev
  et~al\mbox{.}}{2018}]%
        {alev_et_al:LIPIcs:2018:8369}
\bibfield{author}{\bibinfo{person}{Vedat~Levi Alev}, \bibinfo{person}{Nima
  Anari}, \bibinfo{person}{Lap~Chi Lau}, {and} \bibinfo{person}{Shayan~Oveis
  Gharan}.} \bibinfo{year}{2018}\natexlab{}.
\newblock \showarticletitle{{Graph Clustering using Effective Resistance}}. In
  \bibinfo{booktitle}{\emph{9th Innovations in Theoretical Computer Science
  Conference (ITCS)}}, Vol.~\bibinfo{volume}{94}. \bibinfo{pages}{41:1--41:16}.
\newblock


\bibitem[\protect\citeauthoryear{Anari and Gharan}{Anari and Gharan}{2015}]%
        {Anari2015}
\bibfield{author}{\bibinfo{person}{Nima Anari} {and}
  \bibinfo{person}{Shayan~Oveis Gharan}.} \bibinfo{year}{2015}\natexlab{}.
\newblock \showarticletitle{Effective-Resistance-Reducing Flows, Spectrally
  Thin Trees, and Asymmetric {TSP}}. In \bibinfo{booktitle}{\emph{Proceedings
  of the {IEEE} 56th Annual Symposium on Foundations of Computer Science
  (FOCS)}}.
\newblock


\bibitem[\protect\citeauthoryear{Andoni, Krauthgamer, and Pogrow}{Andoni
  et~al\mbox{.}}{2018}]%
        {andoni2018solving}
\bibfield{author}{\bibinfo{person}{Alexandr Andoni}, \bibinfo{person}{Robert
  Krauthgamer}, {and} \bibinfo{person}{Yosef Pogrow}.}
  \bibinfo{year}{2018}\natexlab{}.
\newblock \showarticletitle{On Solving Linear Systems in Sublinear Time}. In
  \bibinfo{booktitle}{\emph{10th Innovations in Theoretical Computer Science
  Conference (ITCS)}}.
\newblock


\bibitem[\protect\citeauthoryear{Banerjee and Lofgren}{Banerjee and
  Lofgren}{2015}]%
        {banerjee2015fast}
\bibfield{author}{\bibinfo{person}{Siddhartha Banerjee} {and}
  \bibinfo{person}{Peter Lofgren}.} \bibinfo{year}{2015}\natexlab{}.
\newblock \showarticletitle{Fast Bidirectional Probability Estimation in Markov
  Models}. In \bibinfo{booktitle}{\emph{Advances in Neural Information
  Processing Systems (NIPS)}}. \bibinfo{pages}{1423--1431}.
\newblock


\bibitem[\protect\citeauthoryear{Borgs, Brautbar, Chayes, and Teng}{Borgs
  et~al\mbox{.}}{2012}]%
        {borgs2012sublinear}
\bibfield{author}{\bibinfo{person}{Christian Borgs}, \bibinfo{person}{Michael
  Brautbar}, \bibinfo{person}{Jennifer Chayes}, {and}
  \bibinfo{person}{Shang-Hua Teng}.} \bibinfo{year}{2012}\natexlab{}.
\newblock \showarticletitle{A sublinear time algorithm for pagerank
  computations}. In \bibinfo{booktitle}{\emph{International Workshop on
  Algorithms and Models for the Web-Graph (WAW)}}. Springer,
  \bibinfo{pages}{41--53}.
\newblock


\bibitem[\protect\citeauthoryear{Bressan, Peserico, and Pretto}{Bressan
  et~al\mbox{.}}{2018}]%
        {bressan2018sublinear}
\bibfield{author}{\bibinfo{person}{Marco Bressan}, \bibinfo{person}{Enoch
  Peserico}, {and} \bibinfo{person}{Luca Pretto}.}
  \bibinfo{year}{2018}\natexlab{}.
\newblock \showarticletitle{Sublinear algorithms for local graph centrality
  estimation}. In \bibinfo{booktitle}{\emph{Proceedings of the IEEE 59th Annual
  Symposium on Foundations of Computer Science (FOCS)}}.
  \bibinfo{pages}{709--718}.
\newblock


\bibitem[\protect\citeauthoryear{Bressan, Peserico, and Pretto}{Bressan
  et~al\mbox{.}}{2019}]%
        {bressan2019approximating}
\bibfield{author}{\bibinfo{person}{Marco Bressan}, \bibinfo{person}{Enoch
  Peserico}, {and} \bibinfo{person}{Luca Pretto}.}
  \bibinfo{year}{2019}\natexlab{}.
\newblock \showarticletitle{On approximating the stationary distribution of
  time-reversible Markov chains}.
\newblock \bibinfo{journal}{\emph{Theory of Computing Systems}}
  (\bibinfo{year}{2019}), \bibinfo{pages}{1--23}.
\newblock


\bibitem[\protect\citeauthoryear{Chandra, Raghavan, Ruzzo, Smolensky, and
  Tiwari}{Chandra et~al\mbox{.}}{1996}]%
        {chandra1996electrical}
\bibfield{author}{\bibinfo{person}{Ashok~K Chandra}, \bibinfo{person}{Prabhakar
  Raghavan}, \bibinfo{person}{Walter~L Ruzzo}, \bibinfo{person}{Roman
  Smolensky}, {and} \bibinfo{person}{Prasoon Tiwari}.}
  \bibinfo{year}{1996}\natexlab{}.
\newblock \showarticletitle{The electrical resistance of a graph captures its
  commute and cover times}.
\newblock \bibinfo{journal}{\emph{Computational Complexity}}
  \bibinfo{volume}{6}, \bibinfo{number}{4} (\bibinfo{year}{1996}),
  \bibinfo{pages}{312--340}.
\newblock


\bibitem[\protect\citeauthoryear{Chiplunkar, Kapralov, Khanna, Mousavifar, and
  Peres}{Chiplunkar et~al\mbox{.}}{2018}]%
        {clusterability}
\bibfield{author}{\bibinfo{person}{Ashish Chiplunkar}, \bibinfo{person}{Michael
  Kapralov}, \bibinfo{person}{Sanjeev Khanna}, \bibinfo{person}{Aida
  Mousavifar}, {and} \bibinfo{person}{Yuval Peres}.}
  \bibinfo{year}{2018}\natexlab{}.
\newblock \showarticletitle{Testing Graph Clusterability: Algorithms and Lower
  Bounds}.
\newblock  (\bibinfo{year}{2018}).
\newblock
\showeprint[arxiv]{1808.04807}


\bibitem[\protect\citeauthoryear{Christiano, Kelner, Madry, Spielman, and
  Teng}{Christiano et~al\mbox{.}}{2011}]%
        {Christiano2011}
\bibfield{author}{\bibinfo{person}{Paul Christiano},
  \bibinfo{person}{Jonathan~A. Kelner}, \bibinfo{person}{Aleksander Madry},
  \bibinfo{person}{Daniel~A. Spielman}, {and} \bibinfo{person}{Shang-Hua
  Teng}.} \bibinfo{year}{2011}\natexlab{}.
\newblock \showarticletitle{Electrical flows, laplacian systems, and faster
  approximation of maximum flow in undirected graphs}. In
  \bibinfo{booktitle}{\emph{Proceedings of the 43rd annual {ACM} Symposium on
  Theory of Computing (STOC)}}. \bibinfo{pages}{273--282}.
\newblock


\bibitem[\protect\citeauthoryear{Chung}{Chung}{1997}]%
        {Chung:1997}
\bibfield{author}{\bibinfo{person}{Fan R.~K Chung}.}
  \bibinfo{year}{1997}\natexlab{}.
\newblock \bibinfo{booktitle}{\emph{Spectral Graph Theory}}.
\newblock \bibinfo{publisher}{American Mathematical Society}.
\newblock


\bibitem[\protect\citeauthoryear{Cohen, Kyng, Miller, Pachocki, Peng, Rao, and
  Xu}{Cohen et~al\mbox{.}}{2014}]%
        {CohenKMPPRX14}
\bibfield{author}{\bibinfo{person}{Michael~B. Cohen}, \bibinfo{person}{Rasmus
  Kyng}, \bibinfo{person}{Gary~L. Miller}, \bibinfo{person}{Jakub~W. Pachocki},
  \bibinfo{person}{Richard Peng}, \bibinfo{person}{Anup~B. Rao}, {and}
  \bibinfo{person}{Shen~Chen Xu}.} \bibinfo{year}{2014}\natexlab{}.
\newblock \showarticletitle{Solving {SDD} linear systems in nearly
  \emph{m}log\({}^{\mbox{1/2}}\)\emph{n} time}. In
  \bibinfo{booktitle}{\emph{Proceedings of the 46th annual {ACM} Symposium on
  Theory of Computing (STOC)}}. \bibinfo{pages}{343--352}.
\newblock


\bibitem[\protect\citeauthoryear{Dubhashi and Panconesi}{Dubhashi and
  Panconesi}{2009}]%
        {dubhashi2009concentration}
\bibfield{author}{\bibinfo{person}{Devdatt~P Dubhashi} {and}
  \bibinfo{person}{Alessandro Panconesi}.} \bibinfo{year}{2009}\natexlab{}.
\newblock \bibinfo{booktitle}{\emph{Concentration of measure for the analysis
  of randomized algorithms}}.
\newblock \bibinfo{publisher}{Cambridge University Press}.
\newblock


\bibitem[\protect\citeauthoryear{Ellens, Spieksma, Van~Mieghem, Jamakovic, and
  Kooij}{Ellens et~al\mbox{.}}{2011}]%
        {ellens2011effective}
\bibfield{author}{\bibinfo{person}{Wendy Ellens}, \bibinfo{person}{FM
  Spieksma}, \bibinfo{person}{P Van~Mieghem}, \bibinfo{person}{A Jamakovic},
  {and} \bibinfo{person}{RE Kooij}.} \bibinfo{year}{2011}\natexlab{}.
\newblock \showarticletitle{Effective graph resistance}.
\newblock \bibinfo{journal}{\emph{Linear algebra and its applications}}
  \bibinfo{volume}{435}, \bibinfo{number}{10} (\bibinfo{year}{2011}),
  \bibinfo{pages}{2491--2506}.
\newblock


\bibitem[\protect\citeauthoryear{Fortunato}{Fortunato}{2010}]%
        {fortunato2010community}
\bibfield{author}{\bibinfo{person}{Santo Fortunato}.}
  \bibinfo{year}{2010}\natexlab{}.
\newblock \showarticletitle{Community detection in graphs}.
\newblock \bibinfo{journal}{\emph{Physics reports}} \bibinfo{volume}{486},
  \bibinfo{number}{3-5} (\bibinfo{year}{2010}), \bibinfo{pages}{75--174}.
\newblock


\bibitem[\protect\citeauthoryear{Hayashi, Akiba, and Yoshida}{Hayashi
  et~al\mbox{.}}{2016}]%
        {hayashi2016efficient}
\bibfield{author}{\bibinfo{person}{Takanori Hayashi}, \bibinfo{person}{Takuya
  Akiba}, {and} \bibinfo{person}{Yuichi Yoshida}.}
  \bibinfo{year}{2016}\natexlab{}.
\newblock \showarticletitle{Efficient algorithms for spanning tree centrality}.
  In \bibinfo{booktitle}{\emph{Proceedings of the 25th International Joint
  Conference on Artificial Intelligence (IJCAI)}}. \bibinfo{pages}{3733--3739}.
\newblock


\bibitem[\protect\citeauthoryear{Jambulapati and Sidford}{Jambulapati and
  Sidford}{2018}]%
        {JS18:sketch}
\bibfield{author}{\bibinfo{person}{Arun Jambulapati} {and}
  \bibinfo{person}{Aaron Sidford}.} \bibinfo{year}{2018}\natexlab{}.
\newblock \showarticletitle{Efficient $\tilde{O}(n/\varepsilon)$ Spectral
  Sketches for the Laplacian and its Pseudoinverse}. In
  \bibinfo{booktitle}{\emph{Proceedings of the 29th Annual ACM-SIAM Symposium
  on Discrete Algorithms (SODA)}}. \bibinfo{pages}{2487--2503}.
\newblock


\bibitem[\protect\citeauthoryear{Jeh and Widom}{Jeh and Widom}{2002}]%
        {jeh2002simrank}
\bibfield{author}{\bibinfo{person}{Glen Jeh} {and} \bibinfo{person}{Jennifer
  Widom}.} \bibinfo{year}{2002}\natexlab{}.
\newblock \showarticletitle{Simrank: a measure of structural-context
  similarity}. In \bibinfo{booktitle}{\emph{Proceedings of the 8th ACM SIGKDD
  International Conference on Knowledge Discovery and Data Mining (KDD)}}.
  \bibinfo{pages}{538--543}.
\newblock


\bibitem[\protect\citeauthoryear{Katz}{Katz}{1953}]%
        {katz1953new}
\bibfield{author}{\bibinfo{person}{Leo Katz}.} \bibinfo{year}{1953}\natexlab{}.
\newblock \showarticletitle{A new status index derived from sociometric
  analysis}.
\newblock \bibinfo{journal}{\emph{Psychometrika}} \bibinfo{volume}{18},
  \bibinfo{number}{1} (\bibinfo{year}{1953}), \bibinfo{pages}{39--43}.
\newblock


\bibitem[\protect\citeauthoryear{Kim and El~Saddik}{Kim and El~Saddik}{2011}]%
        {kim2011personalized}
\bibfield{author}{\bibinfo{person}{Heung-Nam Kim} {and}
  \bibinfo{person}{Abdulmotaleb El~Saddik}.} \bibinfo{year}{2011}\natexlab{}.
\newblock \showarticletitle{Personalized pagerank vectors for tag
  recommendations: inside folkrank}. In \bibinfo{booktitle}{\emph{Proceedings
  of the 5th ACM Conference on Recommender Systems (RecSys)}}.
  \bibinfo{pages}{45--52}.
\newblock


\bibitem[\protect\citeauthoryear{Kunegis and Schmidt}{Kunegis and
  Schmidt}{2007}]%
        {kunegis2007collaborative}
\bibfield{author}{\bibinfo{person}{J{\'e}r{\^o}me Kunegis} {and}
  \bibinfo{person}{Stephan Schmidt}.} \bibinfo{year}{2007}\natexlab{}.
\newblock \showarticletitle{Collaborative filtering using electrical resistance
  network models}. In \bibinfo{booktitle}{\emph{Industrial Conference on Data
  Mining}}. Springer, \bibinfo{pages}{269--282}.
\newblock


\bibitem[\protect\citeauthoryear{Lee, Ozdaglar, and Shah}{Lee
  et~al\mbox{.}}{2013}]%
        {lee2013computing}
\bibfield{author}{\bibinfo{person}{Christina~E Lee}, \bibinfo{person}{Asuman
  Ozdaglar}, {and} \bibinfo{person}{Devavrat Shah}.}
  \bibinfo{year}{2013}\natexlab{}.
\newblock \showarticletitle{Computing the stationary distribution locally}. In
  \bibinfo{booktitle}{\emph{Advances in Neural Information Processing
  Systems}}. \bibinfo{pages}{1376--1384}.
\newblock


\bibitem[\protect\citeauthoryear{Lofgren, Banerjee, and Goel}{Lofgren
  et~al\mbox{.}}{2016}]%
        {lofgren2016personalized}
\bibfield{author}{\bibinfo{person}{Peter Lofgren}, \bibinfo{person}{Siddhartha
  Banerjee}, {and} \bibinfo{person}{Ashish Goel}.}
  \bibinfo{year}{2016}\natexlab{}.
\newblock \showarticletitle{Personalized pagerank estimation and search: A
  bidirectional approach}. In \bibinfo{booktitle}{\emph{Proceedings of the 9th
  ACM International Conference on Web Search and Data Mining (WSDM)}}.
  \bibinfo{pages}{163--172}.
\newblock


\bibitem[\protect\citeauthoryear{Lofgren, Banerjee, Goel, and
  Seshadhri}{Lofgren et~al\mbox{.}}{2014}]%
        {lofgren2014fast}
\bibfield{author}{\bibinfo{person}{Peter~A Lofgren},
  \bibinfo{person}{Siddhartha Banerjee}, \bibinfo{person}{Ashish Goel}, {and}
  \bibinfo{person}{C Seshadhri}.} \bibinfo{year}{2014}\natexlab{}.
\newblock \showarticletitle{FAST-PPR: scaling personalized pagerank estimation
  for large graphs}. In \bibinfo{booktitle}{\emph{Proceedings of the 20th ACM
  SIGKDD International Conference on Knowledge Discovery and Data Mining
  (KDD)}}. \bibinfo{pages}{1436--1445}.
\newblock


\bibitem[\protect\citeauthoryear{Lov{\'a}sz et~al\mbox{.}}{Lov{\'a}sz
  et~al\mbox{.}}{1993}]%
        {lovasz1993random}
\bibfield{author}{\bibinfo{person}{L{\'a}szl{\'o} Lov{\'a}sz} {et~al\mbox{.}}}
  \bibinfo{year}{1993}\natexlab{}.
\newblock \showarticletitle{Random walks on graphs: A survey}.
\newblock \bibinfo{journal}{\emph{Combinatorics, Paul erdos is eighty}}
  \bibinfo{volume}{2}, \bibinfo{number}{1} (\bibinfo{year}{1993}),
  \bibinfo{pages}{1--46}.
\newblock


\bibitem[\protect\citeauthoryear{L{\"u} and Zhou}{L{\"u} and Zhou}{2011}]%
        {lu2011link}
\bibfield{author}{\bibinfo{person}{Linyuan L{\"u}} {and} \bibinfo{person}{Tao
  Zhou}.} \bibinfo{year}{2011}\natexlab{}.
\newblock \showarticletitle{Link prediction in complex networks: A survey}.
\newblock \bibinfo{journal}{\emph{Physica A: statistical mechanics and its
  applications}} \bibinfo{volume}{390}, \bibinfo{number}{6}
  (\bibinfo{year}{2011}), \bibinfo{pages}{1150--1170}.
\newblock


\bibitem[\protect\citeauthoryear{Lyons and Oveis~Gharan}{Lyons and
  Oveis~Gharan}{2017}]%
        {lyons2017sharp}
\bibfield{author}{\bibinfo{person}{Russell Lyons} {and} \bibinfo{person}{Shayan
  Oveis~Gharan}.} \bibinfo{year}{2017}\natexlab{}.
\newblock \showarticletitle{Sharp bounds on random walk eigenvalues via
  spectral embedding}.
\newblock \bibinfo{journal}{\emph{International Mathematics Research Notices}}
  \bibinfo{volume}{2018}, \bibinfo{number}{24} (\bibinfo{year}{2017}),
  \bibinfo{pages}{7555--7605}.
\newblock


\bibitem[\protect\citeauthoryear{Madry, Straszak, and Tarnawski}{Madry
  et~al\mbox{.}}{2014}]%
        {Madry2014}
\bibfield{author}{\bibinfo{person}{Aleksander Madry}, \bibinfo{person}{Damian
  Straszak}, {and} \bibinfo{person}{Jakub Tarnawski}.}
  \bibinfo{year}{2014}\natexlab{}.
\newblock \showarticletitle{Fast Generation of Random Spanning Trees and the
  Effective Resistance Metric}. In \bibinfo{booktitle}{\emph{Proceedings of the
  26th Annual {ACM}-{SIAM} Symposium on Discrete Algorithms (SODA)}}.
\newblock


\bibitem[\protect\citeauthoryear{Nash-Williams}{Nash-Williams}{1959}]%
        {nash1959random}
\bibfield{author}{\bibinfo{person}{C~St~JA Nash-Williams}.}
  \bibinfo{year}{1959}\natexlab{}.
\newblock \showarticletitle{Random walk and electric currents in networks}. In
  \bibinfo{booktitle}{\emph{Mathematical Proceedings of the Cambridge
  Philosophical Society}}, Vol.~\bibinfo{volume}{55}. Cambridge University
  Press, \bibinfo{pages}{181--194}.
\newblock


\bibitem[\protect\citeauthoryear{O{veis Gharan}}{O{veis Gharan}}{2015}]%
        {Shayan_recentadvances}
\bibfield{author}{\bibinfo{person}{Shayan O{veis Gharan}}.}
  \bibinfo{year}{2015}\natexlab{}.
\newblock \showarticletitle{Lecture 4-5: Effective Resistance and Simple Random
  Walks}.
\newblock
  \bibinfo{journal}{\emph{\url{https://homes.cs.washington.edu/~shayan/courses/approx/adv-approx-4.pdf}}}
  (\bibinfo{year}{2015}).
\newblock


\bibitem[\protect\citeauthoryear{Page, Brin, Motwani, and Winograd}{Page
  et~al\mbox{.}}{1999}]%
        {page1999pagerank}
\bibfield{author}{\bibinfo{person}{Lawrence Page}, \bibinfo{person}{Sergey
  Brin}, \bibinfo{person}{Rajeev Motwani}, {and} \bibinfo{person}{Terry
  Winograd}.} \bibinfo{year}{1999}\natexlab{}.
\newblock \bibinfo{booktitle}{\emph{The PageRank citation ranking: Bringing
  order to the web.}}
\newblock \bibinfo{type}{{T}echnical {R}eport}. \bibinfo{institution}{Stanford
  InfoLab}.
\newblock


\bibitem[\protect\citeauthoryear{Ron}{Ron}{2019}]%
        {ron2019sublinear}
\bibfield{author}{\bibinfo{person}{Dana Ron}.} \bibinfo{year}{2019}\natexlab{}.
\newblock \showarticletitle{Sublinear-time algorithms for approximating graph
  parameters}.
\newblock In \bibinfo{booktitle}{\emph{Computing and Software Science}}.
  \bibinfo{publisher}{Springer}, \bibinfo{pages}{105--122}.
\newblock


\bibitem[\protect\citeauthoryear{Sinclair}{Sinclair}{1992}]%
        {sinclair1992improved}
\bibfield{author}{\bibinfo{person}{Alistair Sinclair}.}
  \bibinfo{year}{1992}\natexlab{}.
\newblock \showarticletitle{Improved bounds for mixing rates of Markov chains
  and multicommodity flow}. In \bibinfo{booktitle}{\emph{Latin American
  Symposium on Theoretical Informatics (LATIN)}}. \bibinfo{pages}{474--487}.
\newblock


\bibitem[\protect\citeauthoryear{Song, Cho, Dave, Zhang, and Qiu}{Song
  et~al\mbox{.}}{2009}]%
        {song2009scalable}
\bibfield{author}{\bibinfo{person}{Han~Hee Song}, \bibinfo{person}{Tae~Won
  Cho}, \bibinfo{person}{Vacha Dave}, \bibinfo{person}{Yin Zhang}, {and}
  \bibinfo{person}{Lili Qiu}.} \bibinfo{year}{2009}\natexlab{}.
\newblock \showarticletitle{Scalable proximity estimation and link prediction
  in online social networks}. In \bibinfo{booktitle}{\emph{Proceedings of the
  9th ACM SIGCOMM Conference on Internet Measurement (IMC)}}.
  \bibinfo{pages}{322--335}.
\newblock


\bibitem[\protect\citeauthoryear{Spielman and Srivastava}{Spielman and
  Srivastava}{2011}]%
        {spielman2011graph}
\bibfield{author}{\bibinfo{person}{Daniel~A Spielman} {and}
  \bibinfo{person}{Nikhil Srivastava}.} \bibinfo{year}{2011}\natexlab{}.
\newblock \showarticletitle{Graph sparsification by effective resistances}.
\newblock \bibinfo{journal}{\emph{SIAM J. Comput.}} \bibinfo{volume}{40},
  \bibinfo{number}{6} (\bibinfo{year}{2011}), \bibinfo{pages}{1913--1926}.
\newblock


\bibitem[\protect\citeauthoryear{Yu, Wei, and Berry}{Yu et~al\mbox{.}}{2019}]%
        {haoran2019analyzing}
\bibfield{author}{\bibinfo{person}{Haoran Yu}, \bibinfo{person}{Ermin Wei},
  {and} \bibinfo{person}{Randall~A. Berry}.} \bibinfo{year}{2019}\natexlab{}.
\newblock \showarticletitle{Analyzing Location-Based Advertising for Vehicle
  Service Providers Using Effective Resistances}.
\newblock \bibinfo{journal}{\emph{Proceedings of the ACM on Measurement and
  Analysis of Computing Systems}} \bibinfo{volume}{3}, \bibinfo{number}{1},
  Article \bibinfo{articleno}{6} (\bibinfo{year}{2019}),
  \bibinfo{numpages}{35}~pages.
\newblock


\end{thebibliography}
%\newpage
\appendix
\section{The Chernoff-Hoeffding bound}
We make use of the following Chernoff-Hoeffding bound (see Theorem 1.1 in~\cite{dubhashi2009concentration}).
\begin{theorem}[The Chernoff-Hoeffding bound]\label{thm:chernoff}
	Let $s\geq 1$. Let $X:=\sum_{1\leq i\leq s}X_i$, where $X_i, 1\leq i\leq s$, are independently distributed in $[0,1]$. Then for all $t>0$,
	\[
	\mathbb{P}[|X - \mathbb{E}[X]| > t] \leq e^{-2t^2/s}.
	\]
\end{theorem}
\section{Missing Proofs of Section \ref{sec:algorithms}}\label{sec:app}
We present here the proofs of two theorems in Section \ref{sec:algorithms}. 
\begin{proof}[Proof of Theorem~\ref{thm:esteff-algo-III}]
In Algorithm~\ref{alg:esteff-algo-III}, let $X_i$ be the indicator variable that denotes the $i$-th random walk to be successful (regardless of its length). Then $\mathbb{P}(X_i=1)=R_G(s,t)$. Furthermore, let $X=\sum_1^{M_0}X_i$. Observe that $\mathbb{E}(X)=M_0\cdot R_G(s,t)$.
	
	Next, assume that $R_G(s,t)> \gamma$; let $M_0=\frac{\ln(1/\delta')\cdot 3}{\varepsilon ^2\cdot \gamma}> \frac{\ln(1/\delta')\cdot 3}{\varepsilon ^2\cdot R_G(s,t)}$. Using Chernoff and union bounds we find that
	$\mathbb{P}[|X-\mathbb{E}(X)|>\varepsilon \cdot \mathbb{E}(X)]<2\cdot e^{-\frac{\varepsilon ^2 M_0 R_G(s,t)}{3}}<2\cdot \delta'$ for any $\varepsilon ,\delta' >0$. Thus, we find that with probability at least $1-2\delta'$, $(1-\varepsilon )R_G(s,t)\leq X/M_0\leq (1+\varepsilon )R_G(s,t)$. Now, choosing $\delta'=\frac{\delta}{2}$ yields the desired approximation ratio.

\end{proof}

\begin{proof}[Proof of Theorem \ref{alg:spanningtreelocal}]
	By Lemma~\ref{lemma:local_spantree}, with probability $1-\frac{\delta}{2}$, the $a$ returned in line 1 of Algorithm~\ref{alg:esteff-algo-Ib} satisfies $|a-\frac{\log (T(G'))}{n-1}|\leq \frac{\varepsilon}{2}$. Similarly, with the same probability, $|b-\frac{\log (T(G))}{n}|\leq \frac{\varepsilon}{2}$. By the union bound, this implies that with probability $1-\delta$, $a(n-1)-b n\leq \frac{\varepsilon}{2} (2n-1)-\log (T(G))+\log T((G'))$. Thus, we find that $\frac{e^{a(n-1)}}{e^{bn}}\leq e^{\varepsilon n}\cdot \frac{T(G')}{T(G)}$. Similarly, it holds that $\frac{e^{a(n-1)}}{e^{bn}}\geq e^{-\varepsilon n}\cdot \frac{T(G')}{T(G)}$.
	
	Let $X = \frac{e^{a(n-1)}}{e^{bn}}$. Then by Lemma~\ref{lemma:eff_spantree}, $e^{-\varepsilon n} R_G(s,t)\leq X\leq e^{\varepsilon n} R_G(s,t)$. This yields the desired approximation ratio.
\end{proof}
\end{document}